\documentclass[10pt]{article}
\usepackage{latexsym}
\usepackage{amssymb}
\usepackage{amsthm}
\usepackage{amsmath}
\usepackage{graphicx}
\usepackage[latin1]{inputenc}

\newtheorem{theorem}{Theorem}[section]
\newtheorem{prop}[theorem]{Proposition}
\newtheorem{coro}[theorem]{Corollary}
\newtheorem{lemma}[theorem]{Lemma}
\newtheorem{remark}[theorem]{Remark}

\newcommand{\R}{\mathbb{R}}             
\newcommand{\N}{\mathbb{N}}             
\newcommand{\Z}{\mathbb{Z}}             %
\newcommand{\C}{\mathbb{C}}             
\renewcommand{\H}{\mathcal{H}}          
\newcommand{\B}{\mathcal{B}}            
\newcommand{\M}{\mathcal{M}}            

\newcommand{\e}{\epsilon}
\newcommand{\half}{\frac{1}{2}}

\newcommand{\Ga}{\Gamma^1}         
\newcommand{\Gb}{\Gamma^2}         
\newcommand{\Gc}{\Gamma^3}         

\renewcommand{\d}{\partial}            

\newcommand{\fj}{f_j (X, \lambda,z)}
\newcommand{\fun}{f_1 (X, \lambda,z)}
\newcommand{\funp}{f_1^+ (X, \lambda,z)}

\newcommand{\gj}{g_j (X, \lambda,z)}

\newcommand{\alun}{a_{L1}(\lambda,z)}
\newcommand{\ald}{a_{L2}(\lambda,z)}
\newcommand{\alt}{a_{L3}(\lambda,z)}
\newcommand{\alq}{a_{L4}(\lambda,z)}

\newcommand{\Section}[1]{\section{#1} \setcounter{equation}{0}}

\author{Thierry Daudé \footnote{Department of Mathematics and
    Statistics, McGill University, 805 Sherbrooke South West, Montréal
    QC, H3A 2K6. Email adress: tdaude@math.mcgill.ca} \, and François Nicoleau
    \footnote{Laboratoire Jean Leray, UMR 6629, Université de Nantes, 2, rue de la Houssinière, BP
    92208, 44322 Nantes Cedex 03. Email adress: nicoleau@math.univ-nantes.fr}}
\title{Inverse Scattering at Fixed Energy in de Sitter-Reissner-Nordstr\"om Black Holes}
\date{}

\begin{document}

\maketitle


\begin{abstract}
  In this paper, we consider massless Dirac fields propagating in the outer region of de
  Sitter-Reissner-Nordstr\"om black holes. We show that the metric of such black holes is uniquely
  determined by the partial knowledge of the corresponding scattering matrix $S(\lambda)$ at a fixed energy
  $\lambda \ne 0$. More precisely, we consider the partial wave scattering matrices $S(\lambda,n)$
  (here $\lambda \ne 0$ is the fixed energy and $n \in \N^*$ denotes the angular momentum) defined as the restrictions of the full scattering matrix on a well chosen basis of spin-weighted spherical harmonics. We prove that the mass $M$, the square of the charge $Q^2$ and the cosmological constant $\Lambda$ of a dS-RN black hole (and thus its metric) can be uniquely determined from the knowledge of either the transmission coefficients $T(\lambda, n)$, or the reflexion coefficients $R(\lambda, n)$
  (resp. $L(\lambda, n)$), for all $n \in {\mathcal{L}}$ where $\mathcal{L}$ is a subset of $\N^*$ that satisfies the M\"untz condition $\sum_{n \in {\mathcal{L}}} \frac{1}{n} = +\infty$. Our main tool consists in complexifying the angular momentum $n$ and in studying the analytic properties of the "unphysical" scattering matrix $S(\lambda,z)$ in the complex variable $z$.
  We show in particular that the quantities $\frac{1}{T(\lambda,z)}$, $\frac{R(\lambda,z)}{T(\lambda,z)}$ and $\frac{L(\lambda,z)}{T(\lambda,z)}$ belong to the Nevanlinna class in the region $\{z \in \C, \ Re(z) >0 \}$ for which we have analytic uniqueness theorems at our disposal.  Eventually, as a by-product of our method, we obtain reconstrution formulae for the surface gravities of the event and cosmological horizons of the black hole which have an important physical meaning in the Hawking effect. \\

\vspace{0.5cm}
\end{abstract}


\Section{Introduction}


Black hole spacetimes are among the most fascinating objects whose existence is predicted by Einstein's
General Relativity theory and have attracted most attention in the last decades. From the theoretical point of
view they are simple systems. The only needed parameters for a full description are the mass,
the electric charge, the angular momentum and possibly the cosmological constant of the black hole raising
the natural issue of determining them. From the astrophysical point of view however, black holes are objects
eminently difficult to grasp since they are, by essence, invisible. Only by indirect means can we study some of
their properties and... find actual evidence for their existence! A fruitful approach to better understand their
properties consists in studying how black holes interact with their environment. In particular, it is now well known
established that much can be learned by observing how incoming waves are scattered off a black hole.
We refer for instance to \cite{Ba1,Ba2,Da1,D,DK,HaN,Me1,N} where direct scattering theories for various waves have
been obtained, to \cite{Ba3,Ba4,Ha,Me2} for an application of the previous results to the study of the Hawking effect
and to \cite{Ba5,FKSY} for an analysis of the superradiance phenomenon. In this paper, we follow this general strategy
and address the problem of identifying the metric of a black hole by observing how incoming waves with a given energy
$\lambda$ propagate and scatter at late times. This information is encoded in the scattering matrix $S(\lambda)$
introduced below. More specifically, we shall focus here on the special case of de Sitter-Reissner-Nordstr\"om black
holes and we shall show that the parameters (and thus the metric) of such black holes can be uniquely recovered from the
partial knowledge of the scattering matrix $S(\lambda)$ at a \emph{fixed energy} $\lambda \ne 0$. This is a continuation of our previous
works \cite{DN1,DN2} in which similar questions were addressed and solved from inverse scattering experiments
\emph{at high energies}.



\subsection{de Sitter-Reissner-Nordst\"om black holes}

De Sitter-Reissner-Nordstr\"om (dS-RN) black holes are spherically symmetric electrically charged exact solutions of
the Einstein-Maxwell equations. In Schwarschild coordinates, the exterior region of a dS-RN black hole is described by
the four-dimensional manifold $\M = \mathbb{R}_{t} \times ]r_-,r_+[_r \times S_{\theta,\varphi}^{2}$ equipped with the
lorentzian metric
\begin{equation} \label{Metric}
  g = F(r)\,dt^{2} - F(r)^{-1} dr^{2} - r^{2} \big(d\theta^{2}+\sin^{2}\theta \, d\varphi^{2}\big),
\end{equation}
where
\begin{equation} \label{F}
  F(r) = 1-\frac{2M}{r}+\frac{Q^{2}}{r^{2}} - \frac{\Lambda}{3} r^2.
\end{equation}
The constants $M>0$, $Q \in \R$ appearing in (\ref{F}) are interpreted as the mass and the electric charge of the black
hole and $\Lambda > 0$ is the cosmological constant of the universe. We assume here that the function $F(r)$ has three
simple positive roots $0<r_c<r_-<r_+$ and a negative one $r_n <0$. This is always achieved if we suppose for instance
that $Q^2<\frac{9}{8} M^2$ and that $\Lambda M^2$ be small enough (see \cite{L}). The sphere $\{r=r_c\}$ is called the
Cauchy horizon whereas the spheres $\{r=r_-\}$ and $\{r=r_+\}$ are the event and cosmological horizons respectively.
We shall only consider the exterior region of the black hole, that is the region $\{r_- < r <r_+\}$ lying between the
event and cosmological horizons. Note that the function $F$ is positive there.

The point of view implicitely adopted throughout this work is that of static observers located far from the event and
cosmological horizons of the black hole. We think typically of a telescope on earth aiming at the black hole or at the
cosmological horizon. We understand these observers as living on world lines $\{r = r_0\}$ with $r_- << r_0 << r_+$.
The variable $t$ corresponds to their true perception of time. The event and cosmological horizons which appear as singularities of the metric (\ref{Metric}) are in fact due to our particular choice of coordinates. Using appropriate coordinates system, these horizons can be understood as regular null hypersurfaces that can be crossed one way but would require speeds greater than that of light to be crossed the other way. From the point of view of our observers however, these horizons are thus the boundaries of the \emph{observable} world. This can be more easily understood if we notice that the event and cosmological horizons are in fact never reached in a finite time $t$ by incoming and outgoing radial null geodesics, the trajectories followed by classical light-rays aimed radially at the black hole or at the cosmological horizon. Both horizons are thus perceived as \emph{asymptotic regions} by our static observers.

Instead of working with the radial variable $r$, we make the choice to describe the exterior region of the black hole by using the Regge-Wheeler (RW) radial variable which is more natural when studying the scattering properties of any fields. The RW variable $x$ is defined implicitely by $\frac{dx}{dr} = F^{-1}(r)$, or explicitely by
\begin{equation} \label{RW}
  x = \frac{1}{2\kappa_n} \ln(r-r_n) + \frac{1}{2\kappa_c} \ln(r-r_c) +\frac{1}{2\kappa_-} \ln(r-r_-) + \frac{1}{2\kappa_+} \ln(r_+-r) \ + \ c,
\end{equation}
where $c$ is any constant of integration and the quantities $\kappa_j, \ j=n,c,-,+$ are defined by
\begin{equation} \label{SurfaceGravity}
  \kappa_n = \frac{1}{2} F'(r_n), \ \kappa_c = \frac{1}{2} F'(r_c), \ \kappa_- = \frac{1}{2} F'(r_-), \ \kappa_+ = \frac{1}{2} F'(r_+).
\end{equation}
The constants $\kappa_->0$ and $\kappa_+<0$ are called the surface gravities of the event and cosmological horizons respectively. Note from (\ref{RW}) that the event and cosmological horizons $\{r=r_\pm\}$  are pushed away to the infinities $\{x = \pm \infty\}$ using the RW variable. Let us also emphasize that the incoming and outgoing null radial geodesics become straight lines $\{x=\pm t\}$ in this new coordinates system, a fact that provides a natural manner to define the scattering data simply by mimicking the usual definitions in Minkowski-spacetime. At last, note the presence of a constant of integration $c$ in the definition of $x$. We shall comment on this constant and its consequences on our definition of the scattering matrix below.



\subsection{The scattering matrix and statement of the result}

As waves, we consider massless Dirac fields propagating in the exterior region of a dS-RN black hole. We refer to \cite{Me1,N} for a detailed study of this equation in this background including a complete time-dependent scattering theory. We shall use the expression of the equation obtained in these papers as the starting point of our study. Thus the considered massless Dirac fields are represented by 2 components spinors $\psi$ belonging to the Hilbert space $L^2(\R \times S^2; \, \C^2)$ which satisfy the evolution equation
\begin{equation} \label{FullEq}
  i \partial_{t} \psi = \Big(\Ga D_x + a(x) D_{S^2} \Big) \psi,
\end{equation}
where $\Ga = \textrm{diag}(1,-1)$, $D_x = -i \partial_x$ and $D_{S^2}$ denotes the Dirac operator on $S^2$. Here, the potential $a(x)$ takes the form
\begin{equation} \label{Pot}
  a(x) = \frac{\sqrt{F(r(x))}}{r(x)},
\end{equation}
and thus contains all the information of the metric through the function $F$. In the variable $x$, it will be shown to have the following asymptotics $a(x) \sim \, a_\pm e^{\kappa_\pm x}, \ x \to \pm \infty$ where $a_\pm$ are fixed constants depending on the parameters of the black hole. The equation (\ref{FullEq}) is clearly spherically symmetric and in consequence can be separated. The stationary scattering is thus governed by a countable family of one-dimensional stationnary Dirac equations of the following form
\begin{equation} \label{SE}
  \Big[\Ga D_x - (l+\half) a(x) \Gb \Big] \, \psi(x,\lambda,l) = \lambda \, \psi(x,\lambda,l),
\end{equation}
restrictions of the full stationary equation to a well chosen basis of spin-weighted spherical harmonics (indexed here by $l = \frac{1}{2}, \frac{3}{2},...$) invariant for the full equation. Here $\Ga$ and $\Gb$ are usual $2 \times 2$ Dirac matrices satisfying the anticommutation relations $\Gamma^i \Gamma^j + \Gamma^j \Gamma^i = 2 \delta_{ij}$, $\lambda$ is the energy of the considered waves and $(l+\half), \ l \in \half +\N$ is called the angular momentum. For simplicity, we shall denote the angular momentum $l+\half$ by $n$. Hence the new parameter $n$ runs over the integers $\N$.

As expected thanks to our choice of variable $x$, the stationary equation (\ref{SE}) is a classical one-dimensional massless Dirac equation in flat spacetime perturbed by an exponentially decreasing matrix-valued potential in which the angular momentum $n$ plays the role of a \emph{coupling constant}. Complete stationary scattering theories have been obtained for this type of equation for instance in \cite{AKM,G,HJKS}. Following the approach used in \cite{AKM}, we can thus define in the usual way the scattering matrix $S(\lambda,n)$ in terms of stationary solutions with prescribed asymptotics at infinity, called Jost solutions. These are $2\times2$ matrix-valued functions $F_L$ and $F_R$ solutions of (\ref{SE}) having the asymptotics
\begin{eqnarray*} 
  F_L(x,\lambda,n) & = & e^{i\Ga \lambda x} ( I_2 + o(1)), \ \ x \to +\infty, \\
  F_R(x,\lambda,n) & = & e^{i\Ga \lambda x} (I_2 + o(1)), \ \ x \to -\infty. \nonumber
\end{eqnarray*}
The Jost solutions will be shown to be fundamental matrices of (\ref{SE}). There exists thus a $2\times2$ matrix $A_L(\lambda,n)$ depending only on the energy $\lambda$ and the angular momentum $n$ such that the Jost functions are connected by
$$ 
  F_L(x,\lambda,n) = F_R(x,\lambda,n) A_L(\lambda,n).
$$
The coefficients of the matrix $A_L$ encode all the scattering information of equation (\ref{SE}). In particular, using the notation
\begin{equation} \label{ScatCoef}
  A_L(\lambda,n) = \left[\begin{array}{cc} a_{L1}(\lambda,n)&a_{L2}(\lambda,n)\\a_{L3}(\lambda,n)&a_{L4}(\lambda,n) \end{array} \right],
\end{equation}
the partial wave scattering matrix $S(\lambda,n)$ is then defined by
\begin{equation} \label{SM}
  S(\lambda,n) = \left[ \begin{array}{cc} T(\lambda,n)&R(\lambda,n)\\ L(\lambda,n)&T(\lambda,n) \end{array} \right],
\end{equation}
where
\begin{equation} \label{SMCoeff}
  T(\lambda,n) = a_{L1}^{-1}(\lambda,n), \quad R(\lambda,n) = - \frac{a_{L2}(\lambda,n)}{a_{L1}(\lambda,n)}, \quad L(\lambda,n) = \frac{a_{L3}(\lambda,n)}{a_{L1}(\lambda,n)}.
\end{equation}
The quantities $T$ and $R, L$ are called the transmission and reflection coefficients respectively.
The former measures the part of a signal transmitted from an horizon to the other in a scattering process whereas the
latters measure the part of a signal reflected from an horizon to itself (event horizon for $L$ and cosmological horizon
for $R$)\footnote{Whence the notations $L$ for \emph{left} reflection coefficient since the event horizon is
located "on the left"' at $x=-\infty$ and $R$ for \emph{right} reflection coefficient since the cosmological horizon is located "on the right" at $x=+\infty$.}. At last, the scattering matrix $S(\lambda,n)$ will be shown to be a $2\times2$ unitary matrix.

Roughly speaking the main result of this paper states that either the knowledge of the transmission coefficient
$T(\lambda,n)$ or the knowledge of the reflection coefficients $L(\lambda,n)$ or $R(\lambda,n)$ at a
\emph{fixed energy} $\lambda \ne 0$ and \emph{"for almost all"} $n \in \N$ determines uniquely the mass $M$ and the square
of the charge $Q^2$ of the black hole as well as the cosmological constant $\Lambda$ of the universe. More precisely,
it suffices to know the transmission or reflection
coefficients at a fixed energy $\lambda \ne 0$  on a subset $\mathcal{L} \subset \N^*$ that satisfies the M\"untz condition
$\sum_{n \in \mathcal{L}} \frac{1}{n} = \infty$ in order to prove the uniqueness of the parameters $M, Q^2, \Lambda$. Since the data of the partial wave scattering matrices $S(\lambda,n)$ for all $n \in \N$ is equivalent to know the full
scattering matrix $S(\lambda)$, we can rephrase our main result by: \emph{the partial knowledge of the
scattering matrix $S(\lambda)$ at a fixed energy $\lambda \ne 0$ determines uniquely the metric of a dS-RN black hole}.

Before entering in the description of the method used to proved the above uniqueness result, let us comment on its dependence with our choice of coordinates system. As already mentioned, the variable $x$ is defined by (\ref{RW}) up to a constant of integration $c$. Our definition (\ref{SM})-(\ref{SMCoeff}) of the scattering matrix turns out not to be invariant when we change the constant $c$ in the definition of $x$. More precisely, if $S(\lambda,n)$ denotes the scattering matrix obtained for a given RW variable $x$, we can show that the scattering matrix $\tilde{S}(\lambda)$ obtained using the translated RW variable $\tilde{x} = x + c$ is given by
\begin{equation} \label{NonInv}
  S(\lambda,n) = e^{i\Ga \lambda c} \tilde{S}(\lambda,n) e^{-i\Ga \lambda c},
\end{equation}
or written in components by
\begin{equation} \label{NonInvComp}
  \left[ \begin{array}{cc} T(\lambda,n)& R(\lambda,n)\\ L(\lambda,n)&T(\lambda,n) \end{array} \right] = \left[ \begin{array}{cc} \tilde{T}(\lambda,n)& e^{2i \lambda c} \tilde{R}(\lambda,n)\\ e^{-2i \lambda c} \tilde{L}(\lambda,n)& \tilde{T}(\lambda,n) \end{array} \right].
\end{equation}
Since there is no natural - and better - way to fix the choice of the constant $c$ in (\ref{RW}), we must include the possibility to describe a dS-RN black hole by two different RW variables in the statement of our result. One way to make our result coordinate invariant is to identify the partial wave scattering matrices at a fixed energy $\lambda$ corresponding to all the possible choice of RW variables in the description of a given dS-RN black hole. In other words, we shall say that $S(\lambda,n)$ and $\tilde{S}(\lambda,n)$ are equal when (\ref{NonInv}) or (\ref{NonInvComp}) hold.

Having this in mind, we state now the main uniqueness result of this paper.
\begin{theorem} \label{Main}
  Let $(M,Q,\Lambda)$ and $(\tilde{M},\tilde{Q},\tilde{\Lambda})$ be the parameters of two dS-RN black holes. We denote by $a(x)$ and $\tilde{a}(x)$ the two corresponding potentials appearing in the Dirac equation (\ref{FullEq}). We also denote by $S(\lambda,n)$ and $\tilde{S}(\lambda,n)$ the corresponding partial wave scattering matrices at a fixed energy $\lambda \ne 0$
  defined by (\ref{SM}) and (\ref{SMCoeff}). Consider a subset $\mathcal{L}$ of $\N^*$ that satisfies the M\"untz condition $\sum_{n \in \mathcal{L}} \frac{1}{n} = \infty$ and assume that there exists a constant $c \in \R$ such that one of the following conditions holds:
\begin{eqnarray*}
  (i) & T(\lambda,n) = \tilde{T}(\lambda,n), \quad \forall n \in \mathcal{L},\\
  (ii) & L(\lambda,n) = e^{-2i \lambda c} \tilde{L}(\lambda,n), \quad \forall n \in \mathcal{L},\\
  (iii) & R(\lambda,n) = e^{2i \lambda c} \tilde{R}(\lambda,n), \quad \forall n \in \mathcal{L}.
\end{eqnarray*}
Then the potentials $a$ and $\tilde{a}$ coincide up to translation, \textit{i.e}. there exists a constant $\sigma \in \R$ such that
$$
  a(x) = \tilde{a}(x + \sigma), \quad \forall x \in \R.
$$
As a consequence we get
$$
  M = \tilde{M}, \ Q^2 = \tilde{Q}^2, \ \Lambda = \tilde{\Lambda}.
$$
\end{theorem}
Let us make several comments on this result. \\

1) We emphasize that the uniqueness results in Thm \ref{Main} are in fact twofold. First, we prove that a positive exponentially decreasing potential $a(x)$ satisfying (\ref{AsympA})-(\ref{AsympA'}) for the Dirac equation (\ref{FullEq}) is uniquely determined (up to translation) from one of the assumptions (i)-(iii) of Thm \ref{Main}. Then and only then do we use the particular expression (\ref{Pot}) of this potential to show that the parameters of the black hole (and thus the metric) are uniquely determined. \\

2) As a particular case of Thm \ref{Main}, we see that the potential $a(x)$ and thus the parameters of the black hole are uniquely determined (up to translation for the potential) by the full scattering matrix $S(\lambda)$ at a fixed energy $\lambda \ne 0$. Theorem \ref{Main} is sharp in the sense that the full scattering matrix $S(0)$ at the energy $\lambda = 0$ does not determine uniquely the potential and the parameters (see Remark \ref{S(0)} below). \\


3) In the case of nonzero energies, it is also natural to ask whether the M\"{u}ntz condition $\sum_{n \in \mathcal{L}} \frac{1}{n} = \infty$ is necessary? On one hand, as regards the problem of uniquely determining the parameters of the black hole, it is likely that we could weaken this condition since the metric only depends on "three" parameters. On the other hand, as regards the problem of uniquely determining the potential $a(x)$, the M\"untz condition could be sharp. Indeed a similar inverse scattering problem for 3D Schr\" odinger operators with radial potentials has been already studied by A. G. Ramm in \cite{Ra} and M. Horvatz \cite{Hor}; Ramm showed that the knowledge of a subset of the phase shifts $\delta_l$, with $\sum_{l \in \mathcal{L}} \frac{1}{l} = \infty$, determines uniquely the potential; shortly after, Horv\' ath proved the necessity of the M\"untz condition in some classes of potentials. \\


4) At last, let us say a few words on our uniqueness results from a more geometrical point of view. Notice first that the Dirac equation (\ref{FullEq}) in the exterior region of a dS-RN black hole takes the same form as a Dirac equation on the manifold $\Sigma = \R_x \times S^2$ equipped with the riemanniann metric
\begin{equation} \label{AHMetric}
  g_0 = dx^2 + a^{-2}(x) (d\theta^2 + \sin^2\theta d\varphi^2),
\end{equation}
where $a(x)$ is any smooth positive function. If we assume moreover that the function $a(x)$ has the asymptotics (\ref{AsympA})-(\ref{AsympA'}) as it is the case in our model, then the riemanniann manifold $(\Sigma, g_0)$ can be viewed as a spherically symmetric manifold having two ends $\{x = \pm \infty\}$ that are \emph{asymptotically hyperbolic}. Hence our model fits the more general framework of asymptotically hyperbolic manifolds (AHM). In this setting, Thm \ref{Main} states that metrics like (\ref{AHMetric}) are uniquely determined (up to translations in $x$) from the \emph{partial} knowledge of the scattering matrix $S(\lambda)$ - corresponding to Dirac waves - at a \emph{fixed energy} $\lambda \ne 0$. For more general AHM with no particular symmetry, some direct and inverse scattering results - for \emph{scalar waves} - have been obtained by Joshi, S\'a Barreto in \cite{JSB} and by S\'a Barreto in \cite{SB} (see also \cite{Is2} and \cite{BP}). In \cite{JSB} for instance, it is shown that the asymptotics of the metric of an AHM are uniquely determined (up to diffeomorphisms) by the scattering matrix $S(\lambda)$ at a fixed energy $\lambda$ off a countable subset of $\R$. In \cite{SB}, it is proved that the metric of an AHM is uniquely determined (up to diffeomorphisms) by the scattering matrix $S(\lambda)$ for every $\lambda \in \R \setminus 0$.


\subsection{Overview of the proof}

The main idea of this paper is to complexify the angular momentum $n = l+\half$ and study the analytic properties of the "unphysical" scattering coefficients $T(\lambda,z), L(\lambda,z)$ and $R(\lambda,z)$ (or equivalently the functions $a_{Lj}(\lambda,z)$) with respect to the variable $z \in \C$. The general idea to consider \emph{complex angular momentum} originates in a paper by Regge \cite{Re} as a tool in the analysis of the scattering matrix of Schrodinger operators in $\R^3$ with spherically symmetric potentials. We refer to \cite{Ne}, chapter 13, and \cite{CMo} for a detailed account of this approach. Applications to the study of inverse scattering problems for the same equation can be found in \cite{CMu, CKM,Ra}. These last papers were the starting point of our work.

The first step in our proof of Theorem \ref{Main} relies on uniqueness theorems for analytic functions. Let us define the Nevanlinna class $N(\Pi^+)$ as the set of all analytic functions $f(z)$ on the right half plane $\Pi^+ = \{z \in\C: \ Re(z) >0\}$ that satisfy
$$ 
  \sup_{0<r<1} \int_{-\pi}^{\pi} \ln^+ \Big| f\Big(\frac{1 - re^{i\varphi}}{1+re^{i\varphi}} \Big) \Big| d\varphi < \infty,
$$
where $\ln^+(x) = \left\{ \begin{array}{cc} \ln x, & \ln x \geq 0,\\ 0, & \ln x <0. \end{array} \right.$ Among other properties, it turns out that such functions are uniquely determined by their values on any subset $\mathcal{L} \subset \N^*$ that satisfies the M\"untz condition $\sum_{n \in \mathcal{L}} \frac{1}{n} = \infty$ (see \cite{Ra} and \cite{Ru}, chapter 15, for a more general statement). We shall use the Nevanlinna class and this uniqueness result as follows.

Using explicit representations for the Jost functions as well as the unitarity of the scattering matrix, we first show that the coefficients $a_{Lj}(\lambda,z)$, $j=1,..,4$ in (\ref{ScatCoef}) are entire functions of exponential type in the variable $z$ that satisfy the bound:
$$ 
  |a_{Lj}(\lambda,z)| \leq e^{A|Re(z)|},
$$
where $A$ is the constant given by $\int_\R a(x) dx$. From these estimates, we deduce easily that the functions $a_{Lj}(\lambda,z)$ restricted to the right half plane $\Pi^+$ belong to the Nevanlinna class $N(\Pi^+)$. Hence the preceding discussion allows us to conclude that the functions $a_{Lj}(\lambda,z)$ are completely determined by their values on any subset $\mathcal{L} \subset \N^*$ such that $\sum_{n \in \mathcal{L}} \frac{1}{n} = \infty$.

Since the true scattering data are the transmission coefficient $T$ or the reflection coefficients
$L, R$ and not exactly the $a_{Lj}$, we need to work a bit more to get a useful uniqueness statement.
Using mainly Hadamard's factorization theorem and the previous result, we show that in fact, the whole matrix $A_L(\lambda,z)$ is uniquely determined from the values of one of the scattering coefficients $T(\lambda,n)$, $L(\lambda,n)$ or $R(\lambda,n)$ on any subset $\mathcal{L} \subset \N^*$ such that $\sum_{n \in \mathcal{L}} \frac{1}{n} = \infty$ only.

The second step in our proof relies on precise asymptotics for the coefficients $a_{Lj}(\lambda,z)$ when the parameter $z \to \infty$ for real values of $z$. To obtain these asymptotics, it is convenient to introduce a new radial variable $X$ that has also the great interest to enlight the underlying structure of equation (\ref{SE}). Following \cite{CMu,CKM}, we define the variable $X$ by the Liouville transformation
\begin{equation} \label{LiouvilleIntro}
  X = \int_{-\infty}^x a(s) ds.
\end{equation}
Note that $X$ is well defined thanks to the exponential decay of $a(x)$ at both horizons and runs over the interval $(0,A)$ with $A = \int_\R a(s) ds$. Let us denote by $h(X)$ the inverse transformation of (\ref{LiouvilleIntro}). We shall also use the notations $F_L(X)$ and $F_R(X)$ as a shorthand for the Jost functions $F_L(h(X),\lambda,z)$ and $F_R(h(X),\lambda,z)$. The reason why we introduce such a variable lies in the observation that the components $f_{Lj}(X)$ and $f_{Rj}(X)$ of the Jost functions
$$
F_L(X) = \left[\begin{array}{cc} f_{L1}(X)&f_{L2}(X)\\f_{L3}(X)&f_{L4}(X) \end{array} \right], \quad F_R(X) = \left[\begin{array}{cc} f_{R1}(X)&f_{R2}(X)\\f_{R3}(X)&f_{R4}(X) \end{array} \right],
$$
satisfy then second order differential equations of the form
\begin{equation} \label{SE2ndOrder}
  f''(X) + q(X) f(X) = z^2 f(X), \quad X \in (0,A).
\end{equation}
Here the potential $q$ will be shown to have quadratic singularities at the boundaries $0$ and $A$, \textit{i.e.}
$$
  q(X) \sim \frac{\omega_-}{X^2}, \ X \to 0, \quad \quad \quad q(X) \sim \frac{\omega_+}{(A-X)^2}, \ X \to A,
$$
where $\omega_\pm$ are two constants. We emphasize that the angular momentum or coupling constant $z$ has now become
the spectral parameter of this new equation. We shall show by a perturbative argument that the Jost functions
$f_{Lj}(X)$ and $f_{Rj}(X)$ can be written as small perturbations of certain
\emph{modified Bessel functions}.\footnote{Note that the modified Bessel functions are solutions of the free
equations obtained from (\ref{SE2ndOrder}) by replacing the potential $q(X)$ by its asymptotics at the boundaries,
\textit{i.e.} $f''(X) + \frac{\omega_-}{X^2} f(X) = z^2 f(X)$ when $X \to 0$ and $f''(X) + \frac{\omega_+}{(A-X)^2}
f(X) = z^2 f(X)$ when $X \to A$.} From the well known asymptotics of the modified Bessel functions for large $z$, we
then obtain precise asymptotics for the Jost functions when $z \to \infty$, $z$ real, which in turn immediately yield
the asymptotics of the coefficients $a_{Lj}(\lambda,z)$. Let us remark here that we could also obtain asymptotics of the
scattering data for large values of $z$ in the \emph{complex plane} but we don't need such asymptotics in our proof.
At last, we mention that singular Sturm-Liouville operators like (\ref{SE2ndOrder}) have been studied in details by Freiling and Yurko in \cite{FY} including the solution of various inverse spectral problems.

Let us now explain briefly how we prove our main theorem. Consider two dS-RN black holes with parameters
$M,Q, \Lambda$ and $\tilde{M}, \tilde{Q}, \tilde{\Lambda}$ respectively. We shall use the notation $Z$ and $\tilde{Z}$
for all the relevant scattering quantities relative to these black holes. Assume that one of the conditions
in Thm \ref{Main} holds. From the previous uniqueness result, we conclude first that $A_L(\lambda,z) =
e^{-i\Ga \lambda c} \tilde{A}_L(\lambda,z) e^{i\Ga \lambda c} $ for all $z \in \C$. Second, we use a standard
procedure in one-dimensional inverse spectral problem (see \cite{FY}) together with the
precise asymptotics for the functions $A_L(\lambda,z)$ and $\tilde{A}_L(\lambda,z)$ obtained previously to prove that there exists $k \in \Z$ such that
$$ 
  a(x) = \tilde{a}\Big(x + c + \frac{k \pi}{\lambda}\Big), \ \forall x \in \R.
$$
Hence the potentials $a$ and $\tilde{a}$ coincide up to translations. This proves the first part of Thm \ref{Main}. Using the particular expression of the functions $a$ and $\tilde{a}$ in terms of the parameters of the black holes, we deduce then that
$$
  M = \tilde{M}, \ Q^2 = \tilde{Q}^2, \ \Lambda= \tilde{\Lambda}.
$$
Finally, as a by-product of our method, we also obtain explicit reconstruction formulae for the surface gravities $\kappa_\pm$ from the reflection coefficients $L(\lambda,n)$ or $R(\lambda,n)$ when $n \to \infty$.

This paper is organised as follows. In Section \ref{DirectScattering}, we recall all the direct scattering results of \cite{AKM, Me1, N} useful for the later analysis. In Section \ref{Complexification}, we put together all the analytical results concerning the scattering data $A_L (\lambda,z)$ and $S(\lambda,z)$. In particular, we show there the uniqueness theorems stated above. In Section \ref{AsymptoticsSD}, we introduce the Liouville variable $X$ and calculate the precise asymptotics of the scattering data $A_L(\lambda,z)$ when $z \to \infty$. In Section \ref{Inverse}, we solve the inverse problem and prove our main Theorem Thm \ref{Main}.


\Section{The direct problem} \label{DirectScattering}


In this section, we first recall the expression of the Dirac equation in dS-RN black holes as well as the direct scattering theory obtained in \cite{Me1,N}. We then give an explicit stationary representation of the related scattering matrix following the approach and the notations used in \cite{AKM}.

As explained in the introduction, we describe the exterior region of a dS-RN black hole using the Regge-Wheeler variable $x$ defined in (\ref{RW}). We thus work on the manifold $\B = \R_t \times \Sigma$ with $\Sigma = \R_{x} \times S^2_{\theta, \, \varphi}$, equipped with the metric
$$ 
  g = F(r) (dt^2 - dx^2) -r^2 d\omega^2,
$$
where $F$ is given by (\ref{F}) and $d\omega^2 = d\theta^2 + \sin^2\theta d\varphi^2$ is the euclidean metric induced on $S^2$. The manifold $\mathcal{B}$ is globally hyperbolic meaning that the foliation $\Sigma_t = \{t\} \times \Sigma$ by the level hypersurfaces of the function $t$, is a foliation of $\mathcal{B}$ by Cauchy hypersurfaces (see
\cite{W} for a definition of global hyperbolicity and Cauchy hypersurfaces). In consequence, we can view the
propagation of massless Dirac fields as an evolution equation in $t$ on the spacelike hypersurface $\Sigma$, that is a cylindrical manifold having two distinct ends: $\{x=-\infty\}$ corresponding to the event horizon of the black hole and $\{x=+\infty\}$ corresponding to the cosmological horizon. Let us recall and emphasize here the nature of the geometry - of asymptotically hyperbolic type - of the hypersurface $\Sigma$ near the horizons. This peculiar geometry will be seen in the asymptotic behaviour of the potential $a(x)$ below.


\subsection{Dirac equation and direct scattering results} \label{DiracEq}

We use the form of the massless Dirac equation obtained in \cite{Me1,N}. The fields are represented by a $2$-components spinor belonging to the Hilbert space $\H = L^2(\R \times S^2, dx d\omega; \, \C^2)$ and the evolution equation can be written under Hamiltonian form as
\begin{equation} \label{DiracEquation}
  i \partial_{t} \psi = H \psi,
\end{equation}
where the Hamiltonian $H$ is given by
\begin{equation} \label{FullDiracOperator}
  H =  \Ga D_x + a(x) D_{S^2}.
\end{equation}
The symbol $D_x$ stands for $-i\d_x$ whereas $D_{S^2}$ denotes the Dirac operator on $S^2$ which, in spherical coordinates, takes the form
\begin{equation} \label{DiracSphere}
  D_{S^2} = -i \Gb (\partial_{\theta} + \frac{\cot{\theta}}{2}) - \frac{i}{\sin{\theta}} \Gc   \partial_{\varphi}.
\end{equation}
The potential $a$ is the scalar smooth function given in term of the metric (\ref{Metric})-(\ref{F}) by
\begin{equation} \label{Potential}
  a(x) = \frac{\sqrt{F(r(x))}}{r(x)},
\end{equation}
where $r(x)$ is the inverse diffeomorphism of (\ref{RW}). Finally, the matrices $\Ga, \Gb, \Gc$ appearing in (\ref{FullDiracOperator}) and (\ref{DiracSphere}) are usual $2 \times 2$ Dirac matrices that satisfy the anticommutation relations
$$ 
  \Gamma^i \Gamma^j + \Gamma^j \Gamma^i = 2 \delta_{ij}, \quad \forall i,j=1,2,3.
$$
We shall work with the following representations of the Dirac matrices
$$ 
  \Ga = \left( \begin{array}{cc} 1&0 \\0&-1 \end{array} \right), \quad \Gb = \left( \begin{array}{cc} 0&1 \\1&0 \end{array} \right), \quad \Gc = \left( \begin{array}{cc} 0&i \\-i&0 \end{array} \right).
$$

We use now the spherical symmetry of the equation to simplify further the expression of the Hamiltonian $H$.
We introduce a basis of spin weighted spherical harmonics invariant through the
action of the Hamiltonian that "diagonalize" the Dirac operator $D_{S^2}$.
We refer to I.M. Gel'Fand and Z.Y. Sapiro \cite{GS} for a detailed presentation of these generalized
spherical harmonics and to \cite{Me1,N} for an application to our model. For each spinorial weight
$s$, $2s \in \Z$, we consider the family of spin-weighted spherical harmonics $\{Y^l_{sm} = e^{i m \varphi} u_{sm}^l, \ l-|s| \in \N, \ l - |m| \in \N \}$  which forms a Hilbert basis of $L^2(S^2, d\omega)$ and where the $u^l_{sm}$ satisfy the following relations
\begin{eqnarray*}
  \frac{d u_{sm}^l}{d\theta} - \frac{m - s \cos \theta}{\sin \theta} u_{sm}^l & = & -i[(l+s)(l-s+1)]^\half u_{s-1,m}^l, \\
  \frac{d u_{sm}^l}{d\theta} + \frac{m - s \cos \theta}{\sin \theta} u_{sm}^l & = & -i[(l+s+1)(l-s)]^\half u_{s+1,m}^l.
\end{eqnarray*}
We define the familly $F_m^l = (Y^l_{-\half, m}, Y^l_{\half, m})$ where the indexes $l, m$ belong to $\mathcal{I} = \{(l,m), \ l-\half \in \N, \ l - |m| \in \N \}$. Moreover, we define $\otimes_2$ as the following operation between two vectors of $\C^2$
$$
  \forall u = (u_1, u_2), v=(v_1,v_2), \ u \otimes_2 v = (u_1 v_1, u_2 v_2).
$$
Then the Hilbert space $\H$ can then be decomposed into the infinite direct sum
\begin{displaymath}
  \H = \bigoplus_{(l,m) \in \mathcal{I}} \H_{lm}, \quad \H_{lm} = L^{2}(\R_x; \C^2) \otimes_2 F_{m}^{l}.
\end{displaymath}
We shall henceforth identify $\H_{lm}$ and $L^2(\R; \C^2)$ as well as $\psi_{lm} \otimes_2 F^l_m$ and $\psi_{lm}$. What's more, it is easy to check that the $\H_{lm}$ are let invariant through the action of $H$. Hence we obtain the orthogonal decomposition for the Hamiltonian $H$
\begin{displaymath}
  H = \bigoplus_{(l,m) \in \mathcal{I}} H^{lm},
\end{displaymath}
with
$$ 
  H^{lm} := H_{|\H_{lm}} = \Ga D_x - (l+\half) a(x) \Gb.
$$
In what follows, it will be thus enough to consider the restriction of the Dirac equation (\ref{DiracEquation}) to each Hilbert space $\H_{lm}$ separatly and define there the partial-wave scattering matrices. The full scattering matrix will be then written as the orthogonal sum of these partial-wave scattering matrices.

For later use, we precise here the asymptotics of the potential $a$ in the RW variable defined by (\ref{RW}).
\begin{lemma} \label{AsympPotA}
  \begin{eqnarray}
    a(x) & = & a_\pm e^{\kappa_\pm x} + O(e^{3\kappa_\pm x}), \quad x \to \pm \infty, \label{AsympA}\\
    a'(x) & = & a_\pm \kappa_\pm e^{\kappa_\pm x} + O(e^{3\kappa_\pm x}), \quad x \to \pm \infty, \label{AsympA'}
  \end{eqnarray}
  where the constants $\kappa_\pm$ are given by (\ref{SurfaceGravity}) and $a_\pm$ by
  \begin{equation} \label{apm}
    a_\pm = \frac{\sqrt{\mp 2 \kappa_\pm}}{r_\pm} e^{-\kappa_\pm C_\pm},
  \end{equation}
  with
  \begin{equation} \label{Cpm}
    C_\pm = \ln \Big[ (r_\pm-r_n)^{\frac{1}{2\kappa_n}} (r_\pm-r_c)^{\frac{1}{2\kappa_c}} (r_+-r-)^{\frac{1}{2\kappa_\mp}} \Big] + c,
  \end{equation}
  and $c$ is the constant of integration from (\ref{RW}).
\end{lemma}

\begin{proof}
Since the function $F$ is naturally expressed using the $r$ radial variable, we need to find first equivalents between the RW variable $x$ and $r$ when $r \to r_\pm$ or equivalently when $x \to \pm \infty$. From (\ref{RW}), we have
$$
  x = \frac{1}{2\kappa_\pm} \ln |r - r_\pm| + C_\pm + O(|r-r_\pm|), \quad r \to r_\pm,
$$
where $C_\pm$ are the constants given by (\ref{Cpm}). Hence we get
\begin{equation} \label{Equiv1}
  e^{2 \kappa_\pm (x-C_\pm)} = |r-r_\pm| \Big( 1 + O(|r-r_\pm|) \Big), \quad r \to r_\pm,
\end{equation}
Taking the square root of (\ref{Equiv1}), we also have
\begin{equation} \label{Equiv2}
  e^{\kappa_\pm (x-C_\pm)} = \sqrt{|r-r_\pm|} \Big( 1 + O(|r-r_\pm|) \Big), \quad r \to r_\pm.
\end{equation}
Now a simple calculation shows that
\begin{equation} \label{AsA}
  a(x) = \frac{\sqrt{F(r)}}{r} = \frac{\sqrt{\mp 2 \kappa_\pm}}{r_\pm} \sqrt{|r-r_\pm|} + \ O(|r-r_\pm|^{\frac{3}{2}}), \quad r \to r_\pm.
\end{equation}
Hence (\ref{AsympA}) follows from (\ref{Equiv2}) and (\ref{AsA}). Similarly, we have
$$
  a'(x) = F(r) \frac{d}{dr} \Big(\frac{\sqrt{F(r)}}{r}\Big) = \kappa_\pm \frac{\sqrt{\mp 2 \kappa_\pm}}{r_\pm} \sqrt{|r-r_\pm|} + \ O(|r-r_\pm|^{\frac{3}{2}}), \quad r \to r_\pm,
$$
which together with (\ref{Equiv2}) yields (\ref{AsympA'}).
\end{proof}

In particular, the potential $a(x)$ is thus exponentially decreasing at both horizons. In consequence, the Hamiltonian $H^{lm}$ can be viewed as a very short-range perturbation of the free Hamiltonian $H_0 = \Ga D_x$ on each spin weighted spherical harmonic. Using this fact, the spectral and scattering properties of the Hamiltonians $H, H^{lm}$ are established in an elementary way. The main scattering results obtained in \cite{Me1,N} are summarized in the next proposition
\begin{prop} \label{DS}
  (i) The Hamiltonians $H^{lm}$ and $H$ are selfadjoint on $\H_{lm}$ and $\H$ with domains $D(H^{lm}) = H^1(\R, \C^2)$ and $D(H) = \{ \psi \in \H, \ \psi = \displaystyle\sum_{l,m} \psi_{lm}, \ \psi_{lm} \in D(\H_{lm}), \ \displaystyle\sum_{l,m} \Big( \|H^{lm} \psi_{lm}\|^2 + \|\psi_{lm}\|^2 \Big) < \infty \}$.  \\
  (ii) The Hamiltonians $H^{lm}, H$ have no pure point and singular continuous spectra, \textit{i.e.}
  $$
    \sigma_{pp}(H^{lm}), \sigma_{pp}(H) = \emptyset, \quad \sigma_{sing}(H^{lm}), \sigma_{sing}(H) = \emptyset.
  $$
  In other words, the spectra of $H^{lm}$ and $H$ are purely absolutely continuous. \\
  (iii) Denote by $H_0 = \Ga D_x$ the selfadjoint operator acting on $\H_{lm}$ as well as on $\H$. Then the wave operators $W_{lm}^\pm$ and $W^\pm$ defined by
  $$
    W_{lm}^\pm = s-\lim_{t \to \pm \infty} e^{itH^{lm}} e^{-itH_0}, \quad W^\pm = s-\lim_{t \to \pm \infty} e^{itH} e^{-itH_0},
  $$
  exist and are asymptotically complete.
\end{prop}

As a direct consequence of Proposition \ref{DS}, we can define the partial wave scattering operators $S_{lm}$ by the usual formulae
$$ 
  S_{lm} = (W_{lm}^+)^* W_{lm}^-,
$$
which are well-defined unitary operators on $\H_{lm}$. Let us now introduce the unitary transform on $\H_{lm}$
\begin{equation} \label{F0}
  (F_0 \psi_{lm})(\lambda) = \frac{1}{\sqrt{2\pi}} \int_\R e^{-i\Ga \lambda x} \psi_{lm}(x) dx.
\end{equation}
The transform $F_0$ clearly diagonalizes the free Hamiltonian $H_0 = \Ga D_x$ on $\H_{lm}$. Note in passing that $F_0$ also acts in a trivial way on $\H$. Morover it is unitary and still diagonalizes $H_0$ on $\H$. By definition, the partial wave scattering matrices $S_{lm}(\lambda)$ are simply the partial wave scattering operators $S_{lm}$ written in the energy representation of $H_0$ given by (\ref{F0}), that is for all $\phi_{lm} \in L^2(\R_\lambda,\C^2)$ we have
\begin{equation} \label{PartialSM}
  S_{lm}(\lambda) \phi_{lm}(\lambda) = (F_0 S_{lm} F_0^* \phi_{lm})(\lambda).
\end{equation}
The partial wave scattering matrices $S_{lm}(\lambda)$ are unitary $2\times2$ matrices for all $\lambda \in \R$.

\begin{remark}
  Using the spherical symmetry of the equation and the definitions of the wave operators in Proposition \ref{DS}, the full      scattering operator $S$ and the full scattering matrix $S(\lambda)$ are unitary operators on $\H$ and $L^2(S^2, \C^2)$          defined as the orthogonal sum of the partial wave scattering operators $S_{lm}$ and matrices $S_{lm}(\lambda)$ respectively,    \textit{i.e.}
$$ 
  S = (W^+)^* W^- = \bigoplus_{(l,m) \in \mathcal{L}} S_{lm},
$$
$$ 
  S(\lambda) \phi(\lambda) = (F_0 S F_0^* \phi)(\lambda) = \bigoplus_{(l,m) \in \mathcal{L}} S_{lm}(\lambda) \phi_{lm}(\lambda), \quad \forall \phi \in L^2(\R_\lambda \times S^2, \C^2).
$$
\end{remark}
\begin{remark}
  Observe that the Hamiltonians $H^{lm}$ only depend on the angular momentum $l+\half$, $l \in \half + \N$. From the above definitions, we see immediately that the wave operators $W_{lm}$, the scattering operators $S_{lm}$ and the scattering matrices $S_{lm}(\lambda)$ only depend on the angular momentum $l + \half$ too. For simplicity we shall therefore denote by $S(\lambda,n)$, $ n = l+\half \in \N$, the family of scattering matrices $S_{lm}(\lambda)$.
\end{remark}


\subsection{A stationary representation of the scattering matrix}

In this section, we follow the approach of \cite{AKM} and obtain an explicit stationary representation
of the scattering matrices $S(\lambda,n)$ for a fixed energy $\lambda \in \R$ and all angular momentum
$n \in \N$. Let us emphasize here that the results of this section still hold if we assume
$a \in L^1(\R)$ only.

Let us consider first the stationary solutions of equation (\ref{DiracEquation}) restricted to each spin weighted spherical harmonic, \textit{i.e.} the solutions of
\begin{equation} \label{PartialSE}
  [ \Ga D_x - n a(x) \Gb ] \psi = \lambda \psi, \quad \forall n \in \N.
\end{equation}
Here we can think of $\psi$ in (\ref{PartialSE}) as either a column vector of 2 entries, or as a $2\times2$ matrix. For $\lambda \in \R$, we define the Jost solution from the left $F_L(x,\lambda,n)$ and the Jost solution from the right $F_R(x,\lambda,n)$ as the $2\times2$ matrix solutions of (\ref{PartialSE}) satisfying the following asymptotics
\begin{eqnarray}
  F_L(x,\lambda,n) & = & e^{i\Ga \lambda x} (I_2 + o(1)), \ x \to +\infty, \label{FL}\\
  F_R(x,\lambda,n) & = & e^{i\Ga \lambda x} (I_2 + o(1)), \ x \to -\infty. \label{FR}
\end{eqnarray}
From (\ref{PartialSE}), (\ref{FL}) and (\ref{FR}), it is easy to see that such solutions (if there exist) must satisfy the integral equations
\begin{equation} \label{IE-FL}
  F_L(x,\lambda,n) = e^{i\Ga \lambda x} - i n \Ga \int_x^{+\infty} e^{-i\Ga \lambda (y-x)} a(y) \Gb F_L(y,\lambda,n) dy,
\end{equation}
\begin{equation} \label{IE-FR}
  F_R(x,\lambda,n) = e^{i\Ga \lambda x} + i n \Ga \int_{-\infty}^x e^{-i\Ga \lambda (y-x)} a(y) \Gb F_R(y,\lambda,n) dy.
\end{equation}
Since the potential $a$ belongs to $L^1(\R)$, it follows that the integral equations (\ref{IE-FL}) and (\ref{IE-FR}) are uniquely solvable by iteration and that
$$
  \|F_L(x,\lambda,n)\| \leq e^{n \int_x^{+\infty} a(s) ds}, \quad \|F_R(x,\lambda,n)\| \leq e^{n \int_{-\infty}^x a(s) ds}.
$$
Moreover we can prove
\begin{lemma} \label{DetFL}
  For $\lambda \in \R$ and $n \in \N$, either of the Jost solutions $F_L(x,\lambda,n)$ and $F_R(x,\lambda,n)$ forms a fundamental matrix of (\ref{PartialSE}) and has determinant equal to 1. Moreover, the following equalities hold
\begin{eqnarray}
  F_L(x,\lambda,n)^* \, \Ga \, F_L(x,\lambda,n) & = & \Ga, \label{RelationFL}\\
  F_R(x,\lambda,n)^* \, \Ga \, F_L(x,\lambda,n) & = & \Ga, \label{RelationFR}
\end{eqnarray}
where $^*$ denotes the matrix conjugate transpose.
\end{lemma}
\begin{proof}
See \cite{AKM}, Proposition 2.2.
\end{proof}

Since the Jost solutions are fundamental matrices of (\ref{PartialSE}), there exists a $2\times2$ matrix $A_L(\lambda,n)$ such that $F_L(x,\lambda,n) = F_R(x,\lambda,n) \, A_L(\lambda,n)$. From (\ref{FR}) and (\ref{IE-FL}), we get the following expression for $A_L(\lambda,n)$
\begin{equation} \label{ALRepresentation}
  A_L(\lambda,n) = I_2 - i n \Ga \int_\R e^{-i\Ga \lambda y} a(y) \Gb F_L(y,\lambda,n) dy.
\end{equation}
Moreover, the matrix $A_L(\lambda,n)$ satisfies the following equality (see \cite{AKM}, Proposition 2.2)
\begin{equation} \label{AL-Relation}
  A_L^*(\lambda,n) \Ga A_L(\lambda,n) = \Ga, \quad \forall \lambda \in \R, \ n \in \N.
\end{equation}
Using the notation (\ref{ScatCoef}), the equality (\ref{AL-Relation}) can be written in components as
\begin{equation} \label{ALUnitarity}
  \left. \begin{array}{ccc} |a_{L1}(\lambda,n)|^2 - |a_{L3}(\lambda,n)|^2 & = & 1, \\
  |a_{L4}(\lambda,n)|^2 - |a_{L2}(\lambda,n)|^2 & = & 1, \\
  a_{L1}(\lambda,n) \overline{a_{L2}(\lambda,n)} - a_{L3}(\lambda,n) \overline{a_{L4}(\lambda,n)} & = & 0. \end{array} \right.
\end{equation}
As mentioned in the introduction, the matrices $A_L(\lambda,n)$ encode all the scattering information of equation (\ref{PartialSE}). In particular, it is shown in \cite{AKM} that the scattering matrix $S(\lambda,n)$ defined in (\ref{PartialSM}) has the representation
\begin{equation} \label{SR-SM1}
  S(\lambda,n) = \left[ \begin{array}{cc} T(\lambda,n)&R(\lambda,n)\\ L(\lambda,n)&T(\lambda,n) \end{array} \right],
\end{equation}
where
\begin{equation} \label{SR-SM2}
  T(\lambda,n) = a_{L1}^{-1}(\lambda,n), \quad R(\lambda,n) = - \frac{a_{L2}(\lambda,n)}{a_{L1}(\lambda,n)}, \quad L(\lambda,n) = \frac{a_{L3}(\lambda,n)}{a_{L1}(\lambda,n)}.
\end{equation}
The unitarity of the scattering matrix $S(\lambda,n)$ leads to the following relations
\begin{lemma}
  For each $\lambda \in \R$ and $n \in \N$, we have
  \begin{equation} \label{SCUnitarity}
    \left. \begin{array}{ccc} |T(\lambda,n)|^2 + |R(\lambda,n)|^2 & = & 1, \\
    |T(\lambda,n)|^2 + |L(\lambda,n)|^2 & = & 1, \\
    T(\lambda,n) \overline{R(\lambda,n)} + L(\lambda,n) \overline{T(\lambda,n)} & = & 0. \end{array} \right.
  \end{equation}
\end{lemma}
Note that the relations (\ref{SCUnitarity}) are also direct consequences of the relations (\ref{ALUnitarity}) and definitions (\ref{SR-SM2}).

\begin{remark} \label{ConstantC}
  We finish this section analysing the influence of the constant of integration $c$ used in the definition (\ref{RW}) of the RW variable on the expression of the scattering matrix $S(\lambda,n)$. Assume thus that we describe the same dS-RN black hole using two RW variables $x$ and $\tilde{x} = x + c$. We shall denote by $Z$ and $\tilde{Z}$ all the relevant quantities expressed using the variables $x$ and $\tilde{x}$ respectively. We follow the same procedure as above to define the scattering matrix $\tilde{S}(\lambda,n)$. Our goal is to find a relation between $\tilde{S}(\lambda,n)$ and $S(\lambda,n)$. We start from the stationary equation (\ref{PartialSE}) obtained using $\tilde{x}$. Hence we get
$$ 
  [ \Ga D_{\tilde{x}} - n \tilde{a}(\tilde{x}) \Gb ] \psi = \lambda \psi,
$$
where the potential $\tilde{a}(\tilde{x})$ is simply the translated by $c$ of the potential $a$, \textit{i.e.} $\tilde{a}(\tilde{x}) = a(\tilde{x} - c)$. Therefore, the uniqueness of the Jost functions satisfying the asymptotics (\ref{FL}) and (\ref{FR}) yields
$$ 
  \tilde{F_L}(\tilde{x},\lambda,n) = F_L(\tilde{x} - c,\lambda,n) e^{i \Ga \lambda c}.
$$
Hence, it follows from (\ref{ALRepresentation}) that
$$ 
  \tilde{A_L}(\lambda,n) = e^{-i\Ga \lambda c} A_L(\lambda,n) e^{i\Ga \lambda c}.
$$
Eventually, using (\ref{SR-SM1}) and (\ref{SR-SM2}), we conclude that
$$
  \tilde{S}(\lambda,n) = e^{-i\Ga \lambda c} S(\lambda,n) e^{i\Ga \lambda c} = \left[ \begin{array}{cc} T(\lambda,n)& e^{-2i \lambda c} R(\lambda,n)\\ e^{2i \lambda c} L(\lambda,n)&T(\lambda,n) \end{array} \right],
$$
that is (\ref{NonInv}).
\end{remark}


\Section{Complexification of the angular momentum} \label{Complexification}

In this section, we allow the angular momentum to be complex. After studying the analytic properties of the Jost functions $F_L(x,\lambda,z)$, $F_R(x,\lambda,z)$ and of the matrix $A_L(\lambda,z)$ in the variable $z \in \C$, we prove the uniqueness results mentioned in the introduction.


\subsection{Analytic properties of the Jost functions and matrix $A_L(\lambda,z)$}

Let us start with the Jost functions $F_L(x,\lambda,z)$ and $F_R(x,\lambda,z)$. They are solutions of the stationary equation
\begin{equation} \label{PartialSE1}
  [ \Ga D_x - z a(x) \Gb ] \psi = \lambda \psi, \quad \forall z \in \C.
\end{equation}
with the asymptotics (\ref{FL}) and (\ref{FR}). By commodity, we introduce the Faddeev matrices $M_L(x,\lambda,z)$ and $M_R(x,\lambda,z)$ defined by
\begin{equation} \label{JostFaddeev}
  M_L(x,\lambda,z) = F_L(x,\lambda,z) e^{-i\Ga \lambda x}, \quad M_R(x,\lambda,z) = F_R(x,\lambda,z) e^{-i\Ga \lambda x},
\end{equation}
which thus satisfy the boundary conditions
\begin{eqnarray}
  M_L(x,\lambda,z) & = & I_2 + o(1), \ x \to +\infty, \\
  M_R(x,\lambda,z) & = & I_2 + o(1), \ x \to -\infty.
\end{eqnarray}
We shall also use the notations in components
\begin{equation} \label{Faddeev}
M_L(x,\lambda,z) = \left[\begin{array}{cc} m_{L1}(,x,\lambda,z)&m_{L2}(x,\lambda,z)\\m_{L3}(x,\lambda,z)&m_{L4}(x,\lambda,z) \end{array} \right], \quad M_R(x,\lambda,z) = \left[\begin{array}{cc} m_{R1}(x,\lambda,z)&m_{R2}(x,\lambda,z)\\m_{R3}(x,\lambda,z)&m_{R4}(x,\lambda,z) \end{array} \right].
\end{equation}
From (\ref{IE-FL}) and (\ref{IE-FR}), the Faddeev matrices satisfy the integral equations
\begin{equation} \label{IE-ML}
  M_L(x,\lambda,n) = I_2 - i z \Ga \int_x^{+\infty} e^{-i\Ga \lambda (y-x)} a(y) \Gb M_L(y,\lambda,z) e^{i\Ga \lambda (y-x)} dy,
\end{equation}
\begin{equation} \label{IE-MR}
  M_R(x,\lambda,n) = I_2 + i z \Ga \int_{-\infty}^x e^{-i\Ga \lambda (y-x)} a(y) \Gb M_R(y,\lambda,n) e^{i\Ga \lambda (y-x)}dy.
\end{equation}
Iterating (\ref{IE-ML}) and (\ref{IE-MR}) once, we get the uncoupled systems
\begin{eqnarray}
  m_{L1}(x,\lambda,z) & = & 1 + z^2 \int_x^{+\infty} \int_y^{+\infty} e^{2 i \lambda (t - y)} a(y) a(t) m_{L1}(t,\lambda,z) dt dy, \label{IE-ML1} \\
  m_{L2}(x,\lambda,z) & = & -iz \int_x^{+\infty} e^{-2 i \lambda (y-x)} a(y) dy + z^2 \int_x^{+\infty} \int_y^{+\infty} e^{-2 i \lambda (y-x)} a(y) a(t) m_{L2}(t,\lambda,z) dt dy, \label{IE-ML2}\\
  m_{L3}(x,\lambda,z) & = & iz \int_x^{+\infty} e^{2 i \lambda (y-x)} a(y) dy + z^2 \int_x^{+\infty} \int_y^{+\infty} e^{2 i \lambda (y-x)} a(y) a(t) m_{L3}(t,\lambda,z) dt dy, \label{IE-ML3}\\
  m_{L4}(x,\lambda,z) & = & 1 + z^2 \int_x^{+\infty} \int_y^{+\infty} e^{-2 i \lambda (t - y)} a(y) a(t) m_{L4}(t,\lambda,z) dt dy, \label{IE-ML4}
\end{eqnarray}
and
\begin{eqnarray}
  m_{R1}(x,\lambda,z) & = & 1 + z^2 \int_{-\infty}^x \int_{-\infty}^y e^{-2 i \lambda (y-t)} a(y) a(t) m_{R1}(t,\lambda,z) dt dy, \label{IE-MR1} \\
  m_{R2}(x,\lambda,z) & = & iz \int_{-\infty}^x e^{2 i \lambda (x-y)} a(y) dy + z^2 \int_{-\infty}^x \int_{-\infty}^y e^{2 i \lambda (x - y)} a(y) a(t) m_{R2}(t,\lambda,z) dt dy, \label{IE-MR2}\\
  m_{R3}(x,\lambda,z) & = & -iz \int_{-\infty}^x e^{-2 i \lambda (x-y)} a(y) dy + z^2 \int_{-\infty}^x \int_{-\infty}^y e^{-2 i \lambda (x - y)} a(y) a(t) m_{R3}(t,\lambda,z) dt dy, \label{IE-MR3}\\
  m_{R4}(x,\lambda,z) & = & 1 + z^2 \int_{-\infty}^x \int_{-\infty}^y e^{2 i \lambda (y - t)} a(y) a(t) m_{R4}(t,\lambda,z) dt dy. \label{IE-MR4}
\end{eqnarray}
Iterating the Volterra equations (\ref{IE-ML1})-(\ref{IE-ML4}), we prove easily the following lemma
\begin{lemma} \label{ML-Analytic}
  (i) Set $m_{L1}^0(x,\lambda) = 1$ and $m_{L1}^n(x,\lambda) = \int_x^{+\infty} \int_y^{+\infty} e^{2 i \lambda (t - y)} a(y) a(t) m_{L1}^{n-1}(t,\lambda) dt dy$. Then we get by induction
  $$
    |m_{L1}^n(x,\lambda)| \leq \frac{1}{2n!} \Big( \int_x^{+\infty} a(y) dy \Big)^{2n}.
  $$
  For $x,\ \lambda \in \R$ fixed, the serie $m_{L1}(x,\lambda,z) = \displaystyle\sum_{n=0}^\infty m_{L1}^n(x,\lambda) z^{2n}$
  converges normally on each compact subset of $\C$ and satisfies the estimate
  $$
    |m_{L1}(x,\lambda,z)| \leq \cosh\Big(|z| \int_x^{+\infty} a(s)ds \Big), \quad \forall x \in \R, \ z \in \C .
  $$
  Moreover, the application $z \longrightarrow m_{L1}(x,\lambda,z)$ is entire and even. \\
  (ii) Set $m_{L2}^0(x,\lambda) = -i \int_x^{+\infty} e^{-2 i \lambda (y-x)} a(y) dy$ and $m_{L2}^n(x,\lambda) =
  \int_x^{+\infty} \int_y^{+\infty} e^{2 i \lambda (x - y)} a(y) a(t) m_{L2}^{n-1}(t,\lambda) dt dy$.
  Then we get by induction
  $$
    |m_{L2}^n(x,\lambda)| \leq \frac{1}{(2n+1)!} \Big( \int_x^{+\infty} a(y) dy \Big)^{2n+1}.
  $$
  For $x,\ \lambda \in \R$ fixed, the serie $m_{L2}(x,\lambda,z) = \displaystyle\sum_{n=0}^\infty m_{L2}^n(x,\lambda)
  z^{2n+1}$ converges normally on each compact subset of $\C$ and satisfies the estimate
  $$
    |m_{L2}(x,\lambda,z)| \leq \sinh\Big(|z| \int_x^{+\infty} a(s)ds \Big), \quad \forall x \in \R, \ z \in \C .
  $$
  Moreover, the application $z \longrightarrow m_{L2}(x,\lambda,z)$ is entire and odd.  \\
  (iii) Set $m_{L3}^0(x,\lambda) = i \int_x^{+\infty} e^{2 i \lambda (y-x)} a(y) dy$ and $m_{L3}^n(x,\lambda) =
  \int_x^{+\infty} \int_y^{+\infty} e^{-2 i \lambda (x - y)} a(y) a(t) m_{L3}^{n-1}(t,\lambda) dt dy$.
  Then we get by induction
  $$
    |m_{L3}^n(x,\lambda)| \leq \frac{1}{(2n+1)!} \Big( \int_x^{+\infty} a(y) dy \Big)^{2n+1}.
  $$
  For $x, \ \lambda \in \R$ fixed, the serie $m_{L3}(x,\lambda,z) = \displaystyle\sum_{n=0}^\infty m_{L3}^n(x,\lambda)
  z^{2n+1}$ converges normally on each compact subset of $\C$ and  satisfies the estimate
  $$
    |m_{L3}(x,\lambda,z)| \leq \sinh\Big(|z| \int_x^{+\infty} a(s)ds \Big), \quad \forall x \in \R, \ z \in \C.
  $$
  Moreover, the application $z \longrightarrow m_{L3}(x,\lambda,z)$ is entire and odd. \\
  (iv) Set $m_{L4}^0(x,\lambda) = 1$ and $m_{L4}^n(x,\lambda) = \int_x^{+\infty} \int_y^{+\infty} e^{-2 i \lambda (t - y)} a(y) a(t) m_{L4}^{n-1}(t,\lambda) dt dy$. Then we get by induction
  $$
    |m_{L4}^n(x,\lambda)| \leq \frac{1}{2n!} \Big( \int_x^{+\infty} a(y) dy \Big)^{2n}.
  $$
  For $x,\ \lambda \in \R$ fixed, the serie $m_{L4}(x,\lambda,z) = \displaystyle\sum_{n=0}^\infty m_{L4}^n(x,\lambda)
  z^{2n}$ converges normally on each compact subset of $\C$ and satisfies the estimate
  $$
    |m_{L4}(x,\lambda,z)| \leq \cosh\Big(|z| \int_x^{+\infty} a(s)ds \Big), \quad \forall x \in \R, \ z \in \C .
  $$
  Moreover, the application $z \longrightarrow m_{L4}(x,\lambda,z)$ is entire and even. \\
  (v) Note at last the obvious symmetries
\begin{eqnarray}
  m_{L1}(x,\lambda,z) = \overline{m_{L4}(x,\lambda,\bar{z})}, \quad \forall z \in \C, \\
  m_{L2}(x,\lambda,z) = \overline{m_{L3}(x,\lambda,\bar{z})}, \quad \forall z \in \C.
\end{eqnarray}
\end{lemma}
Of course we have similar results for the Faddeev functions $m_{Rj}(x,\lambda,z), \ j=1,..,4$.
\begin{remark}
  Using the usual notations
  $$
F_L(x,\lambda,z) = \left[\begin{array}{cc} f_{L1}(,x,\lambda,z)&f_{L2}(x,\lambda,z)\\f_{L3}(x,\lambda,z)&f_{L4}(x,\lambda,z) \end{array} \right], \quad F_R(x,\lambda,z) = \left[\begin{array}{cc} f_{R1}(x,\lambda,z)&f_{R2}(x,\lambda,z)\\f_{R3}(x,\lambda,z)&f_{R4}(x,\lambda,z) \end{array} \right],
  $$
and (\ref{JostFaddeev}), we have
\begin{eqnarray}
  f_{Lj}(x,\lambda,z) = e^{i\lambda x} m_{Lj}(x,\lambda,z), & f_{Rj}(x,\lambda,z) = e^{i\lambda x} m_{Rj}(x,\lambda,z),
  & j=1,3, \label{fj-mj1} \\
  f_{Lj}(x,\lambda,z) = e^{-i\lambda x} m_{Lj}(x,\lambda,z), & f_{Rj}(x,\lambda,z) = e^{-i\lambda x} m_{Rj}
  (x,\lambda,z), & j=2,4. \label{fj-mj2}
\end{eqnarray}
Hence the components $f_{Lj}(x,\lambda,z)$ and $f_{Rj}(x,\lambda,z)$ share the same analytic properties as the
Faddeev functions $m_{Lj}(x,\lambda,z)$ and $m_{Rj}(x,\lambda,z)$ respectively.
For later use, we mention two additional properties of the $f_{Lj}(x,\lambda,z)$. First, using (\cite{AKM},
 Prop. 2.2) and the analytic continuation, we have
\begin{equation} \label{determinant}
  det(F_L(x,\lambda,z)) = 1, \quad \forall x \in \R, \ z \in \C.
\end{equation}
Second we notice that the $f_{Lj}(x,\lambda,z)$ and $f_{Rj}(x,\lambda,z)$ satisfy second order differential equations with complex potentials. Precisely, the components $f_{Lj}(x,\lambda,z)$ and $f_{Rj}(x,\lambda,z)$, $j=1,2$ satisfy
\begin{equation} \label{2ndOrder1}
  \Big[ -\frac{d^2}{dx^2} + \frac{a'(x)}{a(x)} \frac{d}{dx} + z^2 a^2(x) - i\lambda \frac{a'(x)}{a(x)} \Big] f = \lambda^2 f,
\end{equation}
whereas the components $f_{Lj}(x,\lambda,z)$ and $f_{Rj}(x,\lambda,z)$, $j=3,4$ satisfy
\begin{equation} \label{2ndOrder2}
  \Big[ -\frac{d^2}{dx^2} + \frac{a'(x)}{a(x)} \frac{d}{dx} + z^2 a^2(x) + i\lambda \frac{a'(x)}{a(x)} \Big] f = \lambda^2 f.
\end{equation}
The differential equations (\ref{2ndOrder1}) and (\ref{2ndOrder2}) follow directly from the uncoupled integral equations (\ref{IE-ML1})-(\ref{IE-MR4}).
\end{remark}

Let us now extend the previous result to the matrix $A_L(\lambda,z)$. From (\ref{ALRepresentation}) and (\ref{Faddeev}), we see first that the components of $A_L(\lambda,z)$ can be expressed by means of the Faddeev functions $m_{Lj}(x,\lambda,z)$ as
\begin{eqnarray}
  a_{L1}(\lambda,z) & = & 1 - i z \int_\R a(x) m_{L3}(x,\lambda,z) dx, \label{al1} \\
  a_{L2}(\lambda,z) & = & - i z \int_\R e^{-2i \lambda x} a(x) m_{L4}(x,\lambda,z) dx, \label{al2} \\
  a_{L3}(\lambda,z) & = & i z \int_\R e^{2i \lambda x} a(x) m_{L1}(x,\lambda,z) dx, \label{al3} \\
  a_{L4}(\lambda,z) & = & 1 + i z \int_\R a(x) m_{L2}(x,\lambda,z) dx. \label{al4}
\end{eqnarray}
Hence we get using Lemma \ref{ML-Analytic}
\begin{lemma} \label{AL-Analytic}
  (i) For $\lambda \in \R$ fixed and all $z \in \C$,
  \begin{eqnarray*}
    a_{L1}(\lambda,z) & = & 1 - i \sum_{n=0}^\infty \Big( \int_\R a(x) m_{L3}^n(x,\lambda) dx \Big) z^{2n+2}, \\
    a_{L2}(\lambda,z) & = & - i \sum_{n=0}^\infty \Big( \int_\R e^{-2i \lambda x} a(x) m_{L4}^n(x,\lambda) dx \Big) z^{2n+1}, \\
    a_{L3}(\lambda,z) & = & i \sum_{n=0}^\infty \Big( \int_\R e^{2i \lambda x} a(x) m_{L1}^n(x,\lambda) dx \Big) z^{2n+1}, \\
    a_{L4}(\lambda,z) & = & 1 + i \sum_{n=0}^\infty \Big( \int_\R a(x) m_{L2}^n(x,\lambda) dx \Big) z^{2n+2}.
  \end{eqnarray*}
  (ii) Set $A = \displaystyle\int_\R a(x) dx$. Then
  \begin{eqnarray}
    |a_{L1}(\lambda,z)|, \ |a_{L4}(\lambda,z)| \leq \cosh(A|z|), \quad \forall z \in \C, \label{AL-ExpType1}\\
    |a_{L2}(\lambda,z)|, \ |a_{L3}(\lambda,z)| \leq \sinh(A|z|), \quad \forall z \in \C. \label{AL-ExpType2}
  \end{eqnarray}
  (iii) The functions $a_{L1}(\lambda,z)$ and $a_{L4}(\lambda,z)$ are entire and even in $z$ whereas the functions $a_{L2}(\lambda,z)$ and $a_{L3}(\lambda,z)$ are entire and odd in $z$. Moreover they satisfy the symmetries
  \begin{eqnarray}
    a_{L1}(\lambda,z) & = & \overline{a_{L4}(\lambda,\bar{z})}, \quad \forall z \in \C, \label{ALSym1}\\
    a_{L2}(\lambda,z) & = & \overline{a_{L3}(\lambda,\bar{z})}, \quad \forall z \in \C. \label{ALSym2}
  \end{eqnarray}
  (iv) The following relations hold for all $z \in \C$
  \begin{eqnarray}
    a_{L1}(\lambda,z) \overline{a_{L1}(\lambda,\bar{z})} - a_{L3}(\lambda,z) \overline{a_{L3}(\lambda,\bar{z})}
    & = & 1, \label{SymAL1-AL3}\\
    a_{L4}(\lambda,z) \overline{a_{L4}(\lambda,\bar{z})} - a_{L2}(\lambda,z) \overline{a_{L2}(\lambda,\bar{z})}
    & = & 1. \label{SymAL2-AL4}
  \end{eqnarray}
\end{lemma}
\begin{proof}
The properties (i)-(iii) are direct consequences of Lemma \ref{ML-Analytic} and of formulae (\ref{al1})-(\ref{al4}). Let us prove (iv). From (\ref{ALUnitarity}), recall that we have for all $z \in \R$,
$$
  |a_{L1}(\lambda,z)|^2 - |a_{L3}(\lambda,z)|^2 = 1.
$$
This last equality can be rewritten as
\begin{equation} \label{AnalyticExt}
  a_{L1}(\lambda,z) \overline{a_{L1}(\lambda,\bar{z})} - a_{L3}(\lambda,z) \overline{a_{L1}(\lambda,\bar{z})} = 1, \quad \forall z \in \R.
\end{equation}
Since the the functions $a_{Lj}(\lambda,z)$ and $\overline{a_{Lj}(\lambda,\bar{z})}$ are entire in
$z$ by (iii), the equality (\ref{AnalyticExt}) extends analytically to the whole complex plane $\C$.
This proves (\ref{SymAL1-AL3}). Using (\ref{ALSym1}), (\ref{ALSym2}) and (\ref{SymAL1-AL3}),
we get (\ref{SymAL2-AL4}).
\end{proof}

At this stage, we have proved that the components of the matrix $A_L(\lambda,z)$
are entire functions of exponential type in the variable $z$.
Precisely, from (\ref{AL-ExpType1}) and (\ref{AL-ExpType2}), we have
\begin{equation} \label{AL-ExpType}
  |a_{Lj}(\lambda,z)| \leq e^{A |z|}, \quad \forall z \in \C, \ j=1,..,4,
\end{equation}
where  ${\displaystyle{A = \int_\R a(x) dx}}$. We now use the relations (\ref{SymAL1-AL3}), (\ref{SymAL2-AL4}) and the parity properties of the $a_{Lj}(\lambda,z)$ to improve this estimate.
\begin{lemma} \label{MainEsti}
  Let $\lambda \in \R$ be fixed. Then for all $z \in \C$
  \begin{equation} \label{MainEst}
    |a_{Lj}(\lambda,z)| \leq e^{A |Re(z)|}, \quad j=1,..,4.
  \end{equation}

\end{lemma}
\begin{proof}
From (\ref{SymAL1-AL3}), we have
$$
  a_{L1}(\lambda,z) \overline{a_{L1}(\lambda,\bar{z})} -
  a_{L3}(\lambda,z) \overline{a_{L3}(\lambda,\bar{z})} = 1, \quad \forall z \in \C.
$$
In particular, we get for purely imaginary $z=iy$
$$
  a_{L1}(\lambda,iy) \overline{a_{L1}(\lambda,-iy)} - a_{L3}(\lambda,iy)
  \overline{a_{L3}(\lambda,-iy)} = 1, \quad \forall y \in \R.
$$
But the parity of $a_{L1}$ and the imparity of $a_{L3}$ yield
\begin{equation} \label{Est-iR1}
  |a_{L1}(\lambda,iy)|^2 + |a_{L3}(\lambda,iy)|^2 = 1, \quad \forall y \in \R.
\end{equation}
Similarly, using (\ref{SymAL2-AL4}) instead of (\ref{SymAL1-AL3}), we get
\begin{equation} \label{Est-iR2}
  |a_{L2}(\lambda,iy)|^2 + |a_{L4}(\lambda,iy)|^2 = 1, \quad \forall y \in \R.
\end{equation}
Hence we conclude from (\ref{Est-iR1}) and (\ref{Est-iR2}) that
\begin{equation} \label{AL-Est-iR}
  |a_{Lj}(\lambda,iy)| \leq 1, \quad \forall y \in \R, \ j=1,..,4.
\end{equation}
Now the estimate (\ref{MainEst}) is a direct consequence of the Phragm\'en-Lindel\"of
theorem (see \cite{Bo}, Thm 1.4.3.) together with (\ref{AL-ExpType}) and (\ref{AL-Est-iR})
as well as the parity properties of the $a_{Lj}(\lambda,z)$.
\end{proof}

For later use, we mention that we have the corresponding estimates for the Jost functions $f_{Lj}(x,\lambda,z)$ and $f_{Rj}(x,\lambda,z)$. Precisely
\begin{lemma} \label{MainEstiF}
  For all $j=1,..,4$ and for all $x \in \R$,
  \begin{eqnarray}
    |f_{Lj}(x,\lambda,z)| \leq C \, e^{|Re(z)| \int_x^{\infty} a(s) ds}, \\
    |f_{Rj}(x,\lambda,z)| \leq C \, e^{|Re(z)| \int_{-\infty}^x a(s) ds}.
  \end{eqnarray}
\end{lemma}
\begin{proof}
Using the relations (\ref{RelationFL}) and (\ref{RelationFR}), we can prove as in the preceding lemma that the Jost functions are bounded on the imaginary axis $z \in i\R$, \textit{i.e.}
$$
  |f_{Lj}(x,\lambda,iy)|, \ |f_{Rj}(x,\lambda,iy)| \leq 1, \quad \forall y \in \R.
$$
We then conclude by the same argument as above.
\end{proof}

We end up this section studying the special case $\lambda=0$ in which explicit calculations can be made. We have
\begin{lemma} \label{Lambda=0}
  Set $X = \int_{-\infty}^x a(s) ds$ and $A = \int_\R a(x) dx$. Then
  \begin{equation} \label{MR-Lambda=0}
    M_R(x,0,z) = \left( \begin{array}{cc} \cosh(zX)& -i\sinh(zX)\\ i\sinh(zX)&\cosh(zX) \end{array} \right), \quad \forall z \in \C,
  \end{equation}
  \begin{equation} \label{ML-Lambda=0}
    M_L(x,0,z) = \left( \begin{array}{cc} \cosh(z(A-X))& -i\sinh(z(A-X))\\ i\sinh(z(A-X))&\cosh(z(A-X)) \end{array} \right), \quad \forall z \in \C.
  \end{equation}
  Moreover
  \begin{equation} \label{AL-Lambda=0}
    A_L(0,z) = \left( \begin{array}{cc} \cosh(Az)& -i\sinh(Az)\\ i\sinh(Az)&\cosh(Az) \end{array} \right), \quad \forall z \in \C.
  \end{equation}
\end{lemma}

\begin{proof}
We just prove a few equalities, for instance $m_{L3}(x,0,z) = -i\sinh(z(A-X))$ and $a_{L1}(0,z) = \cosh(zA)$. From Lemmata \ref{ML-Analytic} and \ref{AL-Analytic} recall that
\begin{eqnarray*}
m_{L3}(x,0,z) & = & \sum_{n=0}^\infty m_{L3}^n(x,0) z^{2n+1}, \\
a_{L1}(0,z) & = & 1 - i \sum_{n=0}^\infty \Big( \int_\R a(x) m_{L3}^n(x,0) dx \Big) z^{2n+2},
\end{eqnarray*}
where
$$
\left\{ \begin{array}{ccc} m_{L3}^0(x,0) & = & i \int_x^{+\infty} a(y) dy, \\ m_{L3}^n(x,0) & = & \int_x^{+\infty} \int_y^{+\infty} a(y) a(t) m_{L3}^{n-1}(t,0) dt dy \end{array} \right.
$$
By induction, we see that
$$
  m_{L3}^n(x,0) = \frac{i}{(2n+1)!} \Big( \int_x^{+\infty} a(s) ds \Big)^{2n+1}, \quad \forall n \in \N.
$$
Therefore, on one hand, we obtain
$$
  m_{L3}(x,0,z) = \sum_{n=0}^\infty m_{L3}(x,0) z^{2n+1} = i \sinh(z(A-X)),
$$
and on the other hand, we get
$$
  a_{L1}(0,z) = 1 + \sum_{n = 0}^\infty \frac{1}{(2n+2)!} A^{2n+2} z^{2n+2} = \cosh(Az).
$$
The other equalities in (\ref{MR-Lambda=0}), (\ref{ML-Lambda=0}) and (\ref{AL-Lambda=0}) are obtained similarly.
\end{proof}

\begin{remark} \label{S(0)}
  It follows from Lemma \ref{Lambda=0} that the partial scattering matrices at the energy $\lambda=0$ are given explicitely by
   $$
      S(0,n) = \left( \begin{array}{cc} \frac{1}{\cosh(nA)}& i\tanh(nA)\\ i\tanh(nA)&\frac{1}{\cosh(nA)} \end{array}
      \right), \quad \forall n \in \N,
   $$
 and thus depend only on the parameter $A$. As a consequence, we conclude that the full scattering matrix $S(0)$
 does not determine uniquely the three parameters of the black hole.
\end{remark}


\subsection{Nevanlinna class and uniqueness results}

In this section, we prove that the coefficients $a_{Lj}(\lambda,z)$ belong to the Nevanlinna class when restricted to the half plane $\Pi^+ = \{z \in\C: \ Re(z) >0\}$. As an application, we prove the uniqueness results mentioned in the introduction.

Recall first that the Nevanlinna class $N(\Pi^+)$ is defined as the set of all analytic functions $f(z)$ on $\Pi^+$ that satisfy the estimate
$$
  \sup_{0<r<1} \int_{-\pi}^{\pi} \ln^+ \Big| f\Big(\frac{1 - re^{i\varphi}}{1+re^{i\varphi}} \Big) \Big| d\varphi < \infty,
$$
where $\ln^+(x) = \left\{ \begin{array}{cc} \ln x, & \ln x \geq 0,\\ 0, & \ln x <0. \end{array} \right.$ We shall use the following result implicit in \cite{Ra}.
\begin{lemma} \label{NevanlinnaCriterion}
  Let $h \in H(\Pi^+)$ be an holomorphic function in $\Pi^+$ satisfying
  \begin{equation} \label{Esti}
   |h(z)| \leq C e^{A \,Re(z)}, \quad \forall z \in \Pi^+,
  \end{equation}
  where $A$ and $C$ are two constants. Then $h \in N(\Pi^+)$.
\end{lemma}

\begin{proof}
Without loss of generality, we can always assume $C \geq 1$. Thus using (\ref{Esti}) and
$\ln^+(ab) \leq \ln^+(a) + \ln^+(b)$ for $a,b>0$, we have
\begin{eqnarray*}
  \ln^+ \Big| h\Big(\frac{1 - re^{i\varphi}}{1+re^{i\varphi}} \Big) \Big| & \leq & \ln C + \ A \,Re \Big(\frac{1 -r e^{i\varphi}}{1+r e^{i\varphi}} \Big), \\
  & \leq & \ln C + \ A \,\frac{1-r^2}{1+r^2 + 2r \cos \varphi}.
\end{eqnarray*}
Now the well known formula $\displaystyle\int_{-\pi}^{\pi} \frac{1-r^2}{1+r^2 + 2r \cos \varphi} = 2 \pi$ yields
the result.
\end{proof}

As a direct consequence of Lemmata \ref{MainEsti} and \ref{NevanlinnaCriterion}, we thus get
\begin{coro} \label{AL-Nevanlinna}
  For each $\lambda \in \R$ fixed, the applications $z \longrightarrow a_{Lj}(\lambda,z)_{|\Pi^+}$ belong to $N(\Pi^+)$.
\end{coro}

Let us recall now a usefull uniqueness theorem involving functions in the Nevanlinna class $N(\Pi^+)$ as stated
in \cite{Ra}. We also refer to \cite{Ru}, Thm 15.23, for a conformly equivalent version of this theorem involving functions in the Nevanlinna class $N(U)$ where $U$ is the unit disc.
\begin{theorem}[\cite{Ra}, Thm 1.3] \label{UniquenessNevanlinna}
  Let $h \in N(\Pi^+)$ satisfying $h(n) = 0$ for all $n \in \mathcal{L}$ where $\mathcal{L} \subset \N^*$ with $\displaystyle\sum_{n \in \mathcal{L}} \frac{1}{n} = \infty$. Then $h \equiv 0$ in $\Pi^+$.
\end{theorem}
Hence we deduce from Corollary \ref{AL-Nevanlinna} and Theorem \ref{UniquenessNevanlinna}
\begin{coro} \label{AL-Uniqueness}
  Consider two dS-RN black holes and denote by $a_{Lj}$ and $\tilde{a}_{Lj}$ the corresponding scattering data. Let $\mathcal{L} \subset \N^*$ satisfying $\displaystyle\sum_{n \in \mathcal{L}} \frac{1}{n} = \infty$. Assume that one of the following equality hold
  $$
    a_{Lj}(\lambda,n) = \tilde{a}_{Lj}(\lambda,n), \quad \forall n \in \mathcal{L}, \quad j=1,..,4.
  $$
  Then
  $$
    a_{Lj}(\lambda,z) = \tilde{a}_{Lj}(\lambda,z), \quad \forall z \in \C, \quad j=1,..,4.
  $$
\end{coro}

\begin{proof}
Apply Theorem \ref{UniquenessNevanlinna} to the difference functions $a_{Lj}(\lambda,z) - \tilde{a}_{Lj}(\lambda,z)$
which belongs to $N(\Pi^+)$ thanks to Lemma \ref{NevanlinnaCriterion}.
\end{proof}

In other words, the scattering data $a_{Lj}(\lambda,z)$ are uniquely determined as functions of $z \in \C$
from their values on the integers $n \in \cal{L}$.
We now extend this uniqueness result assuming that only
the reflection coefficients $L(\lambda,n)$ or $R(\lambda,n)$
are known as stated in Thm \ref{Main}. Precisely, we prove
\begin{prop} \label{Uniqueness}
  Consider two dS-RN black holes and denote by $Z$ and $\tilde{Z}$ all the corresponding scattering data. Let $\mathcal{L} \subset \N^*$ satisfying $\displaystyle\sum_{n \in \mathcal{L}} \frac{1}{n} = \infty$. Assume that there exists a constant $c$ such that one of the following equality hold
  \begin{eqnarray}
    L(\lambda,n) & = & e^{-2i \lambda c} \tilde{L}(\lambda,n), \quad \forall n \in \mathcal{L}, \label{L} \\
    R(\lambda,n) & = & e^{2i \lambda c} \tilde{R}(\lambda,n), \quad \forall n \in \mathcal{L}, \label{R}
  \end{eqnarray}
  Then
  \begin{equation} \label{ALUniqueness}
    \left[ \begin{array}{cc} a_{L1}(\lambda,z)& a_{L2}(\lambda,z)\\ a_{L3}(\lambda,z)&a_{L4}(\lambda,z) \end{array}
    \right] = \left[ \begin{array}{cc} \tilde{a}_{L1}(\lambda,z)& e^{2i \lambda c} \tilde{a}_{L2}(\lambda,z)\\
    e^{-2i \lambda c} \tilde{a}_{L3}(\lambda,z)& \tilde{a}_{L4}(\lambda,z) \end{array} \right], \quad \forall z \in \C.
  \end{equation}
\end{prop}

\begin{proof}
Assume (\ref{L}). Then
\begin{equation} \label{Equa1}
  a_{L3}(\lambda,n) \tilde{a}_{L1}(\lambda,n) = e^{-2i \lambda c} a_{L1}(\lambda,n) \tilde{a}_{L3}(\lambda,n),
  \quad \forall n \in \mathcal{L}.
\end{equation}
By Lemmata \ref{MainEsti} and \ref{NevanlinnaCriterion},
 the product functions $a_{L3}(\lambda,z) \tilde{a}_{L1}(\lambda,z)$ and $a_{L1}(\lambda,z) \tilde{a}_{L3}(\lambda,z)$
 belong to the Nevanlinna class $N$, the equality (\ref{Equa1}) extends analytically to the whole complex plane $\C$. Hence
\begin{equation} \label{Equa2}
  a_{L3}(\lambda,z) \tilde{a}_{L1}(\lambda,z) = e^{-2i \lambda c} a_{L1}(\lambda,z) \tilde{a}_{L3}(\lambda,z), \quad \forall z \in \C.
\end{equation}
But recall from Lemma \ref{AL-Analytic}, (iv), that
$$ 
  a_{L1}(\lambda,z) \overline{a_{L1}(\lambda,\bar{z})} = 1 - a_{L3}(\lambda,z) \overline{a_{L3}(\lambda,\bar{z})}, \quad \forall z \in \C,
$$
from which we deduce that $a_{L1}(\lambda,z)$ and $a_{L3}(\lambda,z)$ have no common zeros.
Hence we infer from (\ref{Equa2}) that the zeros of $a_{Lj}(\lambda,z)$ and $\tilde{a}_{Lj}(\lambda,z)$
for $j=1,3$ coincide with the same multiplicity.

Now recall that the function $a_{L1}(\lambda,z)$ is even. Thus we can write $a_{L1}(\lambda,z) = g(z^2)$ where $g$
is an entire function. Since $a_{L1}(\lambda,z)$ is of order $1$ (\textit{i.e.} $|a_{L1}(\lambda,z) \leq e^{A |z|}$),
we deduce that $g$ is of order $\half$. Hence the Hadamard's factorization theorem, (see \cite{Bo}, Th 2.7.1), yields
$$
  g(\zeta) = G \, \prod_{n=1}^\infty \Big( 1 - \frac{\zeta}{\zeta_n} \Big),
$$
where the $\zeta_n \ne0$ are the zeros of $g$ counted according to
multiplicity, $G = g(0)=a_{L1}(\lambda,0) = 1$ by Lemma \ref{AL-Analytic}, (i). But note that $\zeta_n =
z_n^2$ where the $z_n$ are the zeros of $a_{L1}(\lambda,z)$ by
definition of $g$. Hence we
obtain
$$ 
  a_{L1}(\lambda,z) = \prod_{n=1}^\infty \Big( 1 - \frac{z^2}{z_n^2} \Big).
$$
Similarly, we have
$$
  \tilde{a}_{L1}(\lambda,z) = \prod_{n=1}^\infty \Big( 1 - \frac{z^2}{\tilde{z}_n^2} \Big),
$$
where the $\tilde{z}_n$ are the zeros of $\tilde{a}_{L1}(\lambda,z)$. Since $z_n = \tilde{z}_n$ by the previous discussion, we conclude that
\begin{equation} \label{Equa5}
  a_{L1}(\lambda,z) = \tilde{a}_{L1}(\lambda,z), \quad \forall z \in \C.
\end{equation}
From (\ref{Equa5}), (\ref{Equa2}), (\ref{ALSym1}) and (\ref{ALSym2}), we thus deduce (\ref{ALUniqueness}).
The proof starting from (\ref{R}) is analogous and so we omit it.
\end{proof}

\par
In the same way, if we assume that only the transmission coefficient $T(\lambda,n)$ are known,
we have the following result :

\begin{prop} \label{Uniqueness2}
  Consider two dS-RN black holes and denote by $Z$ and $\tilde{Z}$ all the corresponding scattering data.
  Let $\mathcal{L} \subset \N^*$ satisfying $\displaystyle\sum_{n \in \mathcal{L}} \frac{1}{n} = \infty$.
  Assume that $T(\lambda,n) = \tilde{T}(\lambda,n)$ for all $n \in \mathcal{L}$. Then, there exists a constant $c\in \R$ such  that
  \begin{equation} \label{ALUniqueness2}
    \left[ \begin{array}{cc} a_{L1}(\lambda,z)& a_{L2}(\lambda,z)\\ a_{L3}(\lambda,z)&a_{L4}(\lambda,z) \end{array}
    \right] = \left[ \begin{array}{cc} \tilde{a}_{L1}(\lambda,z)& e^{2i \lambda c} \tilde{a}_{L2}(\lambda,z)\\
    e^{-2i \lambda c} \tilde{a}_{L3}(\lambda,z)& \tilde{a}_{L4}(\lambda,z) \end{array} \right], \quad \forall z \in \C.
  \end{equation}
\end{prop}

\begin{proof}
By Lemma \ref{AL-Analytic}, (iii) and Corollary \ref{AL-Uniqueness}, we have
\begin{equation} \label{Equa6}
  a_{Lj}(\lambda,z) = \tilde{a}_{Lj}(\lambda,z), \quad \forall z \in \C \ j=1,4.
\end{equation}
Now, we use the same strategy as in Prop. \ref{Uniqueness} and we set
${\displaystyle{f(z)= \frac{a_{L3}(\lambda,z)}{z}}}$. Using that $a_{L3}(\lambda,0)=0$, we see that
$f(z)$ is an even entire function of order $1$. Thus, we can write $f(z)=g(z^2)$ where $g$ is an entire
function of order $\frac{1}{2}$. Using the Hadamard's factorization theorem, we obtain as previously
$$
f(z) = G \, z^{2m}\ \prod_{n=1}^\infty \Big( 1 - \frac{z^2}{z_n^2} \Big),
$$
where $2m$ is the multiplicity of $0$, $G$ is a constant and the $z_n$ are the zeros of $f$ counted according to
multiplicity. From (\ref{Equa6}) and Lemma \ref{AL-Analytic}, (iv), we have
$$
 f(z)\overline{f(\bar{z})} = \tilde{f}(z)\overline{\tilde{f}(\bar{z})},
$$
where ${\displaystyle{\tilde{f}(z)= \frac{\tilde{a_{L3}}(\lambda,z)}{z}}}$, or equivalently
$$
  \mid G\mid^2 \ z^{4m} \ \prod_{n=1}^\infty \Big( 1 -\frac{z^2}{z_n^2} \Big)
\Big( 1 -\frac{z^2}{\bar{z_n}^2} \Big)  = \ \mid \tilde{G}\mid^2 \ z^{4\tilde{m}} \ \prod_{n=1}^\infty \Big( 1 -\frac{z^2}{\tilde{z_n}^2} \Big)
\Big( 1 -\frac{z^2}{\bar{\tilde{z_n}}^2} \Big).
$$
It follows that
$$
\mid G \mid =\mid \tilde{G}\mid , \ m=\tilde{m}, z_n=\tilde{z_n}
$$
Recalling that $a_{L3}(\lambda,z)= z f(z)$ and $\tilde{a_{L3}}(\lambda,z)= z \tilde{f}(z)$, we see that there exists a constant $c \in \R$ such that
$$
  a_{L3}(\lambda,z) = e^{-2i \lambda c} \tilde{a}_{L3}(\lambda,z), \quad z \in \C.
$$
Finally, using Lemma \ref{AL-Analytic}, (iii) we obtain the result.
\end{proof}

We finish this section with a first application of the previous uniqueness results to the study of an inverse scattering problem in which the scattering matrix $S(\lambda,n)$ is supposed to be known on an \emph{interval} of energy only (and not simply at a fixed energy $\lambda$). We shall obtain in an elementary way one of the results in \cite{DN2}. Precisely we prove
\begin{coro}
  Assume that one of the scattering coefficients $T(\lambda,n)$, $L(\lambda,n)$ or $R(\lambda,n)$ be known (in the sense of Theorem \ref{Main}) for all $n \in \mathcal{L}$ with $\displaystyle\sum_{n \in \mathcal{L}} \frac{1}{n} = \infty$ and
  on a (possibly small) interval of energy $\lambda \in I$. Then the parameters $M, Q^2, \Lambda$ of the black hole are uniquely determined.
\end{coro}

\begin{proof}
Consider two dS-RN black holes with parameters $M,Q, \Lambda$ and $\tilde{M}, \tilde{Q}, \tilde{\Lambda}$ respectively such that for all $n \in \mathcal{L}$ and for all $\lambda \in I$, one of the following equality holds
  \begin{eqnarray*}
    L(\lambda,n) & = & e^{-2i \lambda c} \tilde{L}(\lambda,n), \\
    R(\lambda,n) & = & e^{2i \lambda c} \tilde{R}(\lambda,n), \\
    T(\lambda,n) & = & \tilde{T}(\lambda,n).
  \end{eqnarray*}
Then, applying Propositions \ref{Uniqueness} and \ref{Uniqueness2}, we see in particular that
$$
  a_{L2}(\lambda,z) = e^{2i \lambda c} \tilde{a}_{L2}(\lambda,z), \quad \forall z \in \C, \ \forall \lambda \in I.
$$
From the first term of the series defining $a_{L2}(\lambda,z)$ and $\tilde{a}_{L2}(\lambda,z)$ (see Lemma \ref{AL-Analytic}, (ii)), we thus obtain
\begin{equation} \label{FourierPot-a}
  \hat{a}(2\lambda) = e^{2i \lambda c} \hat{\tilde{a}}(2\lambda), \quad \forall \lambda \in I,
\end{equation}
where $\hat{a}$ and $\hat{\tilde{a}}$ denote the Fourier transforms of the potentials $a$ and $\tilde{a}$. Since these potentials are exponentially decreasing at both horizons, their Fourier transforms $\hat{a}$ and $\hat{\tilde{a}}$ are analytic in a small strip around the real axis, \textit{i.e.} on $K = \{\lambda \in \C, \ |Im(\lambda)| \leq \e \}$ for $\e$ small enough. Thus the equality (\ref{FourierPot-a}) extends analytically to the whole strip $K$. In particular, we have
$$
  \hat{a}(2\lambda) = e^{2i \lambda c} \hat{\tilde{a}}(2\lambda), \quad \forall \lambda \in \R,
$$
and therefore
\begin{equation} \label{EqualityPotentials}
  a(x) = \tilde{a}(x-c), \quad \forall x \in \R.
\end{equation}
According to the proof following (\ref{UnicitePot}) in section 5, the equality (\ref{EqualityPotentials})
is the essential ingredient to conclude to the uniqueness of the black holes' parameters.
\end{proof}


\Section{Large $z$ asymptotics of the scattering data} \label{AsymptoticsSD}

In this section, we obtain the asymptotic expansion of the scattering data when the coupling constant $z \rightarrow +
\infty$, $z$ real. We emphasize that this is the only place in this paper where we use the exponential decay of the potential $a(x)$ at both horizons and thus the asymptotically hyperbolic nature of the geometry. The main tool to obtain these asymptotics is a simple change of variable $X=g(x)$, called the Liouville transformation. As we shall see in the next section, this Liouville transformation as well as the asymptotics of the scattering data will be useful to study the inverse scattering problem.

As a by-product, we shall obtain a very simple reconstruction formula for the surface gravities $\kappa_{\pm}$ from the scattering reflexion coefficients $L(\lambda, n)$  and $R(\lambda,n),\ n \in \N$. Recall that the surface gravities have an important physical meaning through the celebrated \emph{Hawking effect} which asserts that static observers located far from both horizons measure at late times an isotropic background of thermal radiations emitted from the horizons. The rate of radiation is given by the "temperatures" of the cosmological and event horizons which are shown to be proportional to the surface gravities $\kappa_\pm$.

\subsection{The Liouville transformation.}

We follow the strategy adopted by K. Chadan, R. Kobayashi and M. Musette (see \cite{CMu}, \cite{CKM}). Considering the differential equations (\ref{2ndOrder1}) and (\ref{2ndOrder2}) satisfied by the Jost functions $f_{Lj}(x,\lambda,z)$ and $f_{Rj}(x,\lambda,z)$, we use a Liouville transformation, \textit{i.e.} a simple change of variable $X = g(x)$, that transforms the equations (\ref{2ndOrder1}) and (\ref{2ndOrder2}) into singular Sturm-Liouville differential equations in which the coupling constant $z$ becomes the \emph{spectral parameter} (see Lemma \ref{Sturm} below).

Let us define precisely this Liouville transformation. We denote
\begin{equation} \label{Liouville}
  X= g(x)=\int_{-\infty}^x a(t) \ dt.
\end{equation}
Clearly, $g:\R \rightarrow ]0,A[$ is a $C^1$-diffeomorphism where
\begin{equation} \label{DefA}
  A=\int_{-\infty}^{+\infty} a(t) \ dt.
\end{equation}
For the sake of simplicity, we denote $h=g^{-1}$ the inverse diffeomorphism of $g$ and we use the notation
${\displaystyle{f'(X) = \frac{\partial f}{\partial X} (X)}}$. We also define for $j=1,...,4$, and for $X \in ]0,A[$,
\begin{equation} \label{Fj}
  \fj = f_{Lj} (h(X), \lambda,z).
\end{equation}
\begin{equation} \label{Gj}
  \gj = f_{Rj} (h(X), \lambda,z).
\end{equation}
We begin with an elementary lemma which states that, in the variable $X$, the Jost solutions $\fj$ and $\gj$ satisfy Sturm-Liouville equations with potentials having quadratic singularities at the boundaries.

\begin{lemma} \label{Sturm} \hfill
\begin{enumerate}
  \item For $j=1,2$, $\fj$ and  $\gj$ satisfy on $]0,A[$ the Sturm-Liouville equation
    \begin{equation} \label{SL1}
       y'' +q(X)y = z^2 y.
    \end{equation}
  \item For $j=3,4$, $\fj$ and $\gj$ satisfy on $]0,A[$ the Sturm-Liouville equation
    \begin{equation} \label{SL2}
       y'' +\overline{q(X)}y = z^2 y,
    \end{equation}
    \end{enumerate}
    where the potential ${\displaystyle{q(X) = \lambda^2 h'(X)^2 -i \lambda h''(X)= \frac{\lambda^2}{a^2(x)} +i\lambda \frac{a'(x)}{a^3(x)}}}$ has the asymptotics
    \begin{eqnarray} \label{omega}
       q(X) & = & \frac{\omega_-}{X^2} + \ q_-(X) \ ,\quad \rm{with} \ \ \omega_- = \frac{\lambda^2}{\kappa_-^2} + i \frac{\lambda}{\kappa_-}, \ \ \rm{and} \ \ q_-(X) = O(1), \ X \rightarrow 0. \\
       q(X) & = & \frac{\omega_+}{(A-X)^2} + \ q_+(X) \ , \ \ \rm{with} \ \
\omega_+ = \frac{\lambda^2}{\kappa_+^2} + i \frac{\lambda}{\kappa_+} \ \ \rm{and} \ \ q_+(X) = O(1), \ X \rightarrow A. \label{omega1}
    \end{eqnarray}
\end{lemma}

\begin{proof}
Using (\ref{2ndOrder1}), (\ref{2ndOrder2}) and ${\displaystyle{\frac{dX}{dx} = a(x)}}$, a straightforward calculation gives (\ref{SL1}) and (\ref{SL2}). Let us establish (\ref{omega}). By Lemma \ref{AsympPotA}, recall that the potential $a$ has the asymptotics when $x \to -\infty$
\begin{eqnarray*}
a(x) & = & a_- \ e^{\kappa_- x} + \ O(e^{3\kappa_- x}), \\
a'(x) & = & \kappa_- a_- \ e^{\kappa_- x} + \ O(e^{3\kappa_- x}).
\end{eqnarray*}
When $x \rightarrow - \infty$, or equivalently when $X \rightarrow 0$, it thus follows that
$$ 
X=\int_{-\infty}^x a(t) \ dt = \frac {a_-}{\kappa_-} \ e^{\kappa_- x} + \ O(e^{3\kappa_- x}) = \frac {1}{\kappa_-} \ a(x) + \ O(e^{3\kappa_- x}),
$$
from which we also have
$$
  e^{\kappa_- x} = O(X), \quad X \to 0.
$$
Hence we find
\begin{equation} \label{aX1}
  a(x) = \kappa_- X + O(X^3), \quad X \to 0.
\end{equation}
Simarly we get
\begin{equation} \label{aX2}
  a'(x) = \kappa_-^2 X + O(X^3). \quad X \to 0.
\end{equation}
Now, (\ref{aX1}) and (\ref{aX2}) yield immediately (\ref{omega}). The proof of (\ref{omega1}) is similar.
\end{proof}

In the next lemma, we calculate the wronskians of some pairs of Jost functions in the variable $X$. We recall that the wronskian of two functions $f, g$ is given by $W(f, g) = fg'-f'g$. These wronskians will be useful to obtain the asymptotics of the scattering data and to solve the inverse problem.

\begin{lemma} \label{wronskien} For $z \in \C$, we have :
$$
  W(f_1, f_2)=W(g_1, g_2)= W(f_3, f_4) =W(g_3, g_4) = iz.
$$
\end{lemma}

\begin{proof}
For example, let us calculate $W(f_1, f_2)$. Using ${\displaystyle{\frac{dX}{dx} = a(x)}}$ again, it is clear that
\begin{equation}
W(f_1, f_2)\ =\  \frac{1}{a(x)} W(f_{L1}(x,\lambda,z), \ f_{L2}(x,\lambda,z)).
\end{equation}
Using (\ref{PartialSE1}) and (\ref{determinant}), we obtain easily :
$$
W(f_1, f_2)\ =\ iz (f_{L1}f_{L4} - f_{L2}f_{l3}) = iz\ det\ F_L = iz.
$$
\end{proof}

\subsection{Asymptotics of the Jost functions.}

Singular Sturm Liouville equations like (\ref{SL1}) - (\ref{omega1}) have been yet studied by Freiling and Yurko in \cite{FY}.
In particular, fundamental system of solutions with precise large $z$ asymptotics have been obtained there.
In theory, we could thus establish a dictionary between our Jost functions and these fundamental solutions to infer the asymptotics of
$f_j(X)$ and $g_j(X)$. However, the asymptotics of the fundamental solutions of \cite{FY} are given up to multiplicative constants that,
as regards the inverse problem we have in mind, require to be determined precisely. Since this is not a straightforward task, we prefer to
follow a self-contained and elementary approach that is similar to  the usual perturbative argument for the Sturm Liouville equation
(\ref{SL1})\footnote{A fundamental system of solutions (FSS) of (\ref{SL1})-(\ref{omega1}) could be constructed by perturbation of solutions for the free equation
\begin{equation} \label{Free}
  y'' + \frac{\omega}{X^2} y = z^2 y.
\end{equation}
Noticing that the \emph{modified Bessel functions} form a FSS for (\ref{Free}), we could thus construct a FSS for the full equation using good estimates on the associated Green kernel (see \cite{FY,Is,Se} where a similar procedure can be found). In our approach, we shall retrieve directly that the Jost functions are perturbations of the modified Bessel functions (see Proposition \ref{asymptoticfun}) and thus deduce their asymptotics from the well known asymptotics of the latters.}, but \emph{only} uses the series expansion for the Jost functions obtained in Section \ref{Complexification}. Note that this is not at all compulsory but allows us to keep track of the different multiplicative constants appearing in the perturbative method from the very begining.

Let us study first the Jost function $f_{L1}(x, \lambda,z)$. From Lemma \ref{ML-Analytic} and(\ref{fj-mj1}), we recall that for $z \in \C$,

\begin{equation} \label{jostgauche}
f_{L1}(x, \lambda,z) = e^{i\lambda x}\ \sum_{n=0}^{+\infty} \ m_{L1}^n (x, \lambda) \ z^{2n}
\end{equation}
where
\begin{eqnarray*}
m_{L1}^0 (x,\lambda) &=& 1, \\
m_{L1}^n (x,\lambda) &=&  \int_x^{+\infty} \int_y^{+\infty}\ e^{2i\lambda (t-y)} \ a(y)\ a(t) \ m_{L1}^{n-1} (t, \lambda) \ dt \ dy \ , \
{\rm{for}} \ n \geq 1.
\end{eqnarray*}
So, using the Liouville transformation, we immediately obtain :

\begin{lemma} \label{JostX} \hfill
\begin{equation}\label{funX}
\fun = \sum_{n=0}^{+\infty} a_n (X,\lambda) \ z^{2n},
\end{equation}
where
\begin{equation} \label{anX}
\begin{array}{ccl}
a_0 (X,\lambda) &=& e^{i\lambda h(X)}, \\
a_n (X,\lambda) &=& e^{i\lambda h(X)} \ \int_X^A \int_Y^A e^{-2i\lambda h(Y)} \ e^{i\lambda h(T)} \ a_{n-1} (T, \lambda) \ dT \ dY \ , \ {\rm{for}} \ n \geq 1.
\end{array}
\end{equation}
\end{lemma}

We shall consider $\fun$ as a perturbation of a function $\funp$ where $\funp$ is given by the same serie (\ref{funX})-(\ref{anX}) as $\fun$ with $h(X) = g^{-1}(X)$ replaced by another diffeomorphism denoted $h^+(X) = g_+^{-1}(X)$. The diffeomorphism $g_+$ in turn is defined in the same manner as $g$ but we replace $a(x)$ by its equivalent at $+\infty$. More precisely, if we write $A-g(x) = \int_x^{+\infty}  a(t) \ dt$, it is natural to set for $X \in ]0,A[$
$$
A-g_+(x) = \int_x^{+\infty}  a_+ \ e^{\kappa_+ t} \ dt = -\frac{a_+}{\kappa_+} e^{\kappa_+ x}.
$$
So, we define
\begin{equation}\label{newdiffeo}
h_+ (X)= g_+^{-1} (X) =  \frac{1}{\kappa_+} \ \log (A-X) + \,C_+,
\end{equation}
with
\begin{equation}\label{cstplus}
C_+ =  \frac{1}{\kappa_+} \log (-\frac{\kappa_+}{a_+}).
\end{equation}
The function $\funp$ is thus given by
\begin{equation} \label{funplus}
\funp = \sum_{n=0}^{+\infty} a_n^+ (X,\lambda) \ z^{2n},
\end{equation}
where
\begin{eqnarray*}
a_0^+ (X,\lambda) &=& e^{i\lambda h^+(X)}, \\
a_n^+ (X,\lambda) &=& e^{i\lambda h^+(X)} \ \int_X^A \int_Y^A e^{-2i\lambda h^+(Y)} \ e^{i\lambda h^+(T)} \ a_{n-1}^+ (T, \lambda) \ dT \ dY \ , \
{\rm{for}} \ n \geq 1.
\end{eqnarray*}

Thanks to our choice of diffeomorphism $h^+$, the coefficients of the serie (\ref{funplus}) can be explicitely calculated. Precisely, denoting by $\Gamma$ the well-known Gamma function, we have

\begin{lemma} \label{Jostplus}
For $X\in ]0,A[$, $z \in \C$ and for all $n \geq 0$
$$
a_n^+ (X,\lambda) = (-\frac{\kappa_+}{a_+})^{\frac{i\lambda}{\kappa_+}} \ \Gamma(1-\nu_+)\ {1 \over
{2^{2n}\ \Gamma(n+1-\nu_+ )\ n!}} \ (A-X)^{2n+\frac{i\lambda}{\kappa_+}}
$$
with
\begin{equation} \label{nuplus}
\nu_+ = \frac{1}{2} - i \frac{\lambda}{\kappa_+}
\end{equation}
\end{lemma}

\begin{proof}
We prove the formula by induction. For $n=0$, the result is clear by (\ref{newdiffeo}), (\ref{cstplus}).
For $n\geq 1$, an elementary calculation gives
\begin{equation} \label{formedev}
a_n^+ (X,\lambda) = (-\frac{\kappa_+}{a_+})^{\frac{i\lambda}{\kappa_+}} \ \frac{1}{(1+\frac{2i\lambda}{\kappa_+}) \cdots
(2n-1+\frac{2i\lambda}{\kappa_+}) \ 2 \cdots(2n)} \ (A-X)^{2n+\frac{i\lambda}{\kappa_+}}
\end{equation}
Using the functional equality $\Gamma(z+1) = z \Gamma(z)$, Lemma \ref{Jostplus} is proved.
\end{proof}

Now, it turns out that the expressions for the coefficients $a_n^+(X,\lambda)$ can be written in terms of the modified Bessel function $I_{-\nu}(x)$. Let us recall its definition (\cite{Le}, Eq. $(5.7.1)$, p. $108$),
\begin{equation} \label{Besselmod}
I_{-\nu}(x) =\sum_{n=0}^{+\infty}\  \frac{1}{\Gamma(n-\nu +1)\  n!}\ \left(\frac{x}{2}\right)^{-\nu+2n}, \ \ x \in \C, \ \mid Arg\ x \mid < \pi.
\end{equation}
We deduce
\begin{coro} \label{asjostp} \hfill
\begin{enumerate}
\item For $X \in ]0,A[$ and $z \in \C$,
\begin{equation}\label{jostplusfinal}
\funp = (-\frac{\kappa_+}{a_+})^{\frac{i\lambda}{\kappa_+}} \ \Gamma(1-\nu_+) \ \sqrt{A-X} \ (\frac{z}{2})^{\nu_+} \ I_{-\nu_+} (z(A-X)).
\end{equation}
\item Let $X_1 \in ]0,A[$ fixed. Then, for $k=0,1$ and for all $X \in ]0,X_1[$, the following asymptotics hold when $z \rightarrow + \infty$, $z$ real:
\begin{equation} \label{asymjostplus}
 f_1^{+(k)}(X,\lambda, z)= (-1)^k\ \frac{2^{-\nu_+}}{\sqrt{2\pi}}\ (-\frac{\kappa_+}{a_+})^{\frac{i\lambda}{\kappa_+}} \ \Gamma(1-\nu_+)
\ z^{k-\frac{i\lambda}{\kappa_+}} \ e^{z(A-X)} \ \Big(1+O(\frac{1}{z})\Big).
\end{equation}
\end{enumerate}
\end{coro}
\textit{Proof}: The first assertion comes from Lemma (\ref{Jostplus}) and (\ref{Besselmod}),
observing that ${\displaystyle{(\frac{x}{2})^{\nu} \ I_{-\nu}(x)}}$ is holomorphic on $\C$. For the second one, let us recall the well-known asymptotics for the modified Bessel function $I_{-\nu}(x)$, $\nu \in \C$, $k=0,1$, when $x \rightarrow +\infty$ :
\begin{equation}\label{asymptbessel}
I_{-\nu}^{(k)}(x) = \frac{e^x}{\sqrt{2\pi x}} \ (1 + O(\mid x\mid^{-1}))\ \ ,\ \ x \rightarrow +\infty.
\end{equation}
For the case $k=0$, we refer to (\cite{Le}, Eq. $(5.11.8)$, p. $123$). The case $k=1$ follows from the previous case together with the identity (see \cite{Le}, Eq. $(5.7.9)$, p. $110$) :
$$ 
2 \ I_{-\nu} ' (x) = I_{-\nu-1}(x) + I_{-\nu+1}(x).
$$
Then, the asymptotics (\ref{asymjostplus}) are a simple consequence of (\ref{asymptbessel}). $\diamondsuit$

\begin{remark}
Let us study the special case $\lambda=0$. From the definitions, it is immediate to see that $f_1(X,0,z)=f_1^+(X,0,z)$ and $\nu_+ = \frac{1}{2}$. So, Corollary \ref{asjostp} entails that
\begin{equation}\label{lambdazero}
\fun =  \Gamma(\frac{1}{2})  \ {\sqrt{\frac{z(A-X)}{2}}}\ I_{-\frac{1}{2}} (z(A-X)).
\end{equation}
Using that ${\displaystyle{\Gamma(\frac{1}{2}) = \sqrt{\pi}}}$ and the equality (see \cite{Le}, Eq. $(5.8.5)$, p. $112$),
\begin{equation}\label{Iundemi}
I_{-\frac{1}{2}} (x) = \sqrt{\frac{2}{\pi x}}\ \cosh x \ ,
\end{equation}
we get $f_1(X,0,z) = \cosh(z(A-X))$. Hence we rediscover the result obtained in Lemma \ref{Lambda=0}.
\end{remark}

In order to estimate $\fun -\funp$, we have to control $h(X)-h_+(X)$. Since in the construction of $h_+(X)$, we have replaced $a(x)$ by its asymptotic at $+\infty$, it is hopeless to get globally estimates on $]0,A[$. This is why we shall work in $]X_0, A[$ where $X_0 \in ]0,A[$ is fixed. We have the following result :

\begin{lemma}\label{estdiffeo}
Let $X_0 \in]0,A[$ fixed. Then, there exists $C>0$ such that for $k =0,1,2$
\begin{equation}\label{estimdiffeo}
\mid h^{(k)}(X)-h_+^{(k)}(X) \mid \ \leq\ C\ (A-X)^{2-k}\ \ ,\ \ \forall X \in ]X_0,A[.
\end{equation}
\end{lemma}

\begin{proof}
We only give the proof for $k=0$ since the other cases are similar. For $X\in ]X_0,A[$, by Lemma \ref{AsympPotA} we have
\begin{equation}\label{prem}
A-X = \int_x^{+\infty} a(t)\ dt = \int_x^{+\infty} \left( a_+ e^{\kappa_+ t} +O( e^{3\kappa_+ t}) \right) \ dt =
-\frac{a_+}{\kappa_+} e^{\kappa_+ x}
+O( e^{3\kappa_+ x}).
\end{equation}
Then,
$$
\log (-\frac{\kappa_+}{a_+}) + \log (A-X) = \kappa_+ x +O(e^{2\kappa_+ x}).
$$
By (\ref{prem}), we see that $O(e^{2\kappa_+ x})= O((A-X)^2)$. So, using (\ref{cstplus}), we obtain :
$$
x=h(X)= \frac{1}{\kappa_+}  \log (A-X) + C_+ + O((A-X)^2) = h_+ (X) + O((A-X)^2).
$$
\end{proof}

We can deduce from Lemma \ref{estdiffeo} some useful properties for $a_n(X, \lambda)$ and  $a_n^+(X, \lambda)$.

\begin{lemma} \label{Estimations}\hfill
\begin{enumerate}
\item For $n \geq 0$ and $X \in]0,A[$, we have
\begin{equation}\label{estimAn}
\mid a_n (X,\lambda) \mid \leq \frac {(A-X)^{2n}}{(2n)!},
\end{equation}
\item For $n \geq 0$ and $X \in]X_0,A[$,
\begin{equation}\label{estimAn'}
a_n' (X,\lambda) = O\left((A-X)^{2n-1}\right).
\end{equation}
\end{enumerate}
The same estimates hold for $a_n^+(X,\lambda)$.
\end{lemma}

\begin{proof}
The first point is clear by a simple induction. Let us prove the second one. For $n=0$, we observe that $a_0 '(X,\lambda) = i h'(X)\  a_0 (X)$. But, since ${\displaystyle{h_+ '(X) = -\frac{1}{\kappa_+ (A-X)}}}$, Lemma \ref{estdiffeo} yields the estimate
\begin{equation} \label{h'(X)}
 {\displaystyle{h'(X) = O\left( \frac{1}{A-X}\right)}}, \quad \forall X \in ]X_0,A[,
\end{equation}
which proves the result. For $n\geq1$, we have
$$
a_n (X,\lambda) = i\lambda  \ h'(X)\ a_n (X) - e^{-i\lambda h(X)}\  \int_X^A e^{-i\lambda h(T)}\ a_{n-1} (T,\lambda) \ dT.
$$
Hence the result follows from (\ref{estimAn}) and (\ref{h'(X)}) by induction. For $a_n^+(X,\lambda)$,
the proof is identical.
\end{proof}

Now, we want to control the difference $\fun- \funp$. To do this, we set
\begin{equation}\label{error}
e_n(X,\lambda) =a_n(X, \lambda) - a_n^+(X, \lambda),
\end{equation}
and thus, we have
\begin{equation}\label{difference}
\fun- \funp = \sum_{n=0}^{+\infty} \ e_n(x,\lambda) \  z^{2n}.
\end{equation}
In the next lemma, we show that $a_n(X,\lambda), \ a_n^+(X,\lambda)$ and $e_n(X,\lambda)$ satisfy second order differential equations.
\begin{lemma}\label{eqdiffs} \hfill
\begin{enumerate}
\item For $n\geq 1$, $a_n(X,\lambda)$ and $a_n^+(X,\lambda)$ satisfy on $]0,A[$,
\begin{eqnarray}
a_n''(X, \lambda) + q(X) \ a_n (X, \lambda) & = & a_{n-1}(X, \lambda), \label{eqdiffAn}\\
{a_n^+}''(X, \lambda) + q_+(X) \ a_n^+ (X, \lambda) &=& a_{n-1}^+(X, \lambda), \label{eqdiffAn+}
\end{eqnarray}
where
\begin{equation}\label{qplus}
q_+(X) = \lambda^2 {h_+}'(X)^2 -i \lambda {h_+}''(X) = \left( \frac{\lambda^2}{\kappa_+^2} +\frac{i\lambda}{\kappa_+}\right) \ \frac{1}{(A-X)^2}.
\end{equation}
\item For $n\geq 1$, $e_n(X,\lambda)$ satisfies on $]0,A[$,
\begin{equation}\label{eqdiffen}
{e_n}''(X, \lambda) + q_+ (X)   \ e_n (X, \lambda) =
e_{n-1}(X, \lambda) - q_0(X) \ a_n (X, \lambda)
\end{equation}
where
\begin{equation}\label{qzero}
q_0(X) = q(X)- q_+(X).
\end{equation}
\end{enumerate}
\end{lemma}

\begin{proof}
Since ${\displaystyle{\fun= \sum_{n=0}^{+\infty} \ a_n(x,\lambda) \  z^{2n}}}$, (\ref{eqdiffAn}) follows directly
from (\ref{SL1}). The proof of (\ref{eqdiffAn+}) is identical. At last, (\ref{eqdiffen}) is a direct consequence
of (\ref{eqdiffAn}) and (\ref{eqdiffAn+}).
\end{proof}

Now we show that the equation (\ref{eqdiffen}) can be rewritten as an integral equation which will be useful to estimate the error term $e_n$.

\begin{lemma}\label{equationerreur} \hfill
\begin{enumerate}
\item $q_0(X) \in L^{\infty} (X_0, A)$.
\item For $n\geq 1$ and for all $X \in ]X_0,A[$, $e_n(X,\lambda)$ satisfies  the integral equation
\begin{equation}\label{eqintegrale}
e_n (X,\lambda) = e^{i\lambda h_+ (X)} \ \int_X^A \int_Y^A e^{-2i\lambda h_+ (Y)}\  e^{i\lambda h_+ (T)} \
[ e_{n-1}(T,\lambda) -q_0 (T)\ a_n (T,\lambda) ] \  dT\ dY.
\end{equation}
\end{enumerate}
\end{lemma}

\begin{proof}
The first point is clear by Lemma \ref{estdiffeo} after having noticed that
\begin{equation}\label{decomp}
q_0(X) = \lambda^2 (h'(X)-{h_+} '(X))\ (h'(X)+{h_+} '(X)) -i\lambda (h''(X)- {h_+}''(X))
\end{equation}
To prove the second point, we denote by $f_n (X,\lambda)$ the (R.H.S) of (\ref{eqintegrale}). Using (\ref{eqdiffen}), we have :
$$
f_n (X,\lambda) = e^{i\lambda h_+ (X)} \ \int_X^A \int_Y^A e^{-2i\lambda h_+ (Y)}\  e^{i\lambda h_+ (T)} \ {e_n}''(T, \lambda)\  dT\ dY
\hspace{3cm}
$$
$$
  + \ e^{i\lambda h_+ (X)} \ \int_X^A \int_Y^A e^{-2i\lambda h_+ (Y)}\  e^{i\lambda h_+ (T)} \
    q_+ (T)   \ e_n (T, \lambda)     \  dT\ dY.
$$
Integrating by part twice the first integral and using Lemma \ref{Estimations} yield (\ref{eqintegrale}).
\end{proof}

In the next lemma, we estimate $e_n(X,\lambda)$ and its derivative.

\begin{lemma}\label{estimationerreur}
There exists a constant $C>0$ such that for all $n \geq 0$ and for all $X \in ]X_0, A[$,
\begin{equation}\label{estimen}
\mid e_n(X,\lambda) \mid \ \leq \ C \ (n+1) \ \frac {(A-X)^{2n+2}}{(2n+2)!}
\end{equation}
\begin{equation}\label{estimend}
\mid {e_n}'(X,\lambda) \mid \ \leq \ C \ (n+1)   \  \frac {(A-X)^{2n+1}}{(2n+1)!}
\end{equation}
\end{lemma}

\begin{proof}
We prove the lemma by induction. For $n=0$, we have
$$
e_0(x,\lambda) = e^{i\lambda h(X)}-e^{i\lambda h_+(X)} = e^{i\lambda h_+(X)} \left( e^{i\lambda (h(X)-h_+(X))} -1 \right).
$$
So, by Lemma \ref{estdiffeo}, we obtain for a suitable constant $C>0$,
$$
\mid e_0(X,\lambda) \mid \leq \frac{C}{2} \ (A-X)^2, \quad \forall X \in (X_0,A).
$$
Assuming that the property is true for $n-1$, we have by  (\ref{estimAn}) and (\ref{eqintegrale}),
$$
\mid e_n(X,\lambda)\mid \leq \int_X^A \int_Y^A \left( C\ n \ \frac{(A-T)^{2n}}{(2n)!} + C \ \frac{(A-T)^{2n}}{(2n)!} \right) \ dT\ dY, \quad \forall X \in (X_0,A),
$$
where we have supposed that ${\displaystyle{C \ \geq \ \mid\mid q_0 \mid\mid_{L^\infty (X_0,A)}}}$. So,
$$
\mid e_n(X,\lambda) \mid \leq C \ (n+1) \frac {(A-X)^{2n+2}}{(2n+2)!}, \quad \forall X \in (X_0,A).
$$
We prove (\ref{estimend}) similarly.
\end{proof}

Now, we can establish the main result of this section :

\begin{prop}\label{asymptoticfun}\hfill
\begin{enumerate}
\item There exists $C>0$ such that for $k=0,1$, for all $X \in ]X_0,A[$ and all $z >0$,
\begin{equation}\label{differenceAsymp}
\mid f_1^{(k)}(X,\lambda, z) - f_1^{+ (k)} (X,\lambda, z) \mid \ \leq \ C \ (A-X) \ z^{k-1} \ e^{z(A-X)}.
\end{equation}
\item For fixed $X_0<X_1$ with $X_j \in ]0, A[$, $\forall k=0,1$, $\forall X\in]X_0,X_1[$, we have the following asymptotics,
when $z \rightarrow +\infty$,
\begin{equation} \label{asymflun}
 f_1^{(k)}(X,\lambda, z)= (-1)^k\ \frac{2^{-\nu_+}}{\sqrt{2\pi}}\ (-\frac{\kappa_+}{a_+})^{\frac{i\lambda}{\kappa_+}} \ \Gamma(1-\nu_+)
\ z^{k-\frac{i\lambda}{\kappa_+}} \ e^{z(A-X)} \ \Big(1+O(\frac{1}{z})\Big).
\end{equation}
\end{enumerate}
\end{prop}

\begin{proof}
We only prove (\ref{differenceAsymp}) in the case $k=0$ since the case $k=1$ is similar. By  Lemma \ref{estimationerreur} and for $z>0$, we have :
\begin{eqnarray*}
\mid \fun -\funp \mid \ & \leq & \ \sum_{n=0}^{+\infty} \mid e_n (x,\lambda)\mid \ z^{2n}\\
                        & \leq & \ \sum_{n=0}^{+\infty} C\ (n+1) \ \frac {(A-X)^{2n+2}}{(2n+2)!}\ z^{2n}\\
                        & \leq & \ \frac{C}{z} \ \ (A-X) \ \sum_{n=0}^{+\infty}  \frac{(n+1)}{2n+2} \ \frac {(z(A-X))^{2n+1}}{(2n+1)!} \\
                        & \leq & \ \frac{C}{2z} \ (A-X) \ \sinh(z(A-X)) \\
                        & \leq & \ \frac{C}{z} \ (A-X) \  e^{z(A-X)}.
\end{eqnarray*}
Now, since $z$ is real, (\ref{asymflun}) follows from Corollary \ref{asjostp} and (\ref{differenceAsymp}).
\end{proof}

In order to calculate the asymptotics of $f_2(X,\lambda,z)$, we follow the same method as for $f_1(X,\lambda,z)$. We thus only give the final results omitting the details. First, we construct $f_2^+(X,\lambda,z)$ which approximates $f_2(X,\lambda,z)$ as in (\ref{differenceAsymp}). We obtain :

\begin{equation}\label{jostplusdeux}
 f_2^+(X,\lambda,z) = -i\ (-\frac{\kappa_+}{a_+})^{-\frac{i\lambda}{\kappa_+}} \ \Gamma(1-\mu_+) \ \sqrt{A-X} \ (\frac{z}{2})^{\mu_+} \
 I_{1-\mu_+} (z(A-X)),
\end{equation}
where
\begin{equation}\label{muplus}
\mu_+ = \frac{1}{2} +i \frac{\lambda}{\kappa_+}\ .
\end{equation}
Then, using the well-known asymptotics for the modified Bessel functions, we deduce :

\begin{prop}\label{asymptoticfdeux} For fixed $X_0<X_1$ with $X_j \in ]0, A[$, $\forall k=0,1$, $\forall X\in]X_0,X_1[$,
we have the following asymptotics,
when $z \rightarrow +\infty$,
\begin{equation} \label{asymfldeux}
 f_2^{(k)}(X,\lambda, z)= (-1)^{k+1}\ i\ \frac{2^{-\mu_+}}{\sqrt{2\pi}}\ (-\frac{\kappa_+}{a_+})^{-\frac{i\lambda}{\kappa_+}} \ \Gamma(1-\mu_+)
\ z^{k+\frac{i\lambda}{\kappa_+}} \ e^{z(A-X)} \ \Big(1+O(\frac{1}{z})\Big).
\end{equation}
\end{prop}

\begin{remark}
  As previously, let us study the special case $\lambda=0$. We have $f_2 (X,0,z) = f_2^+ (X,0,z)$ and $\mu_+ = \frac{1}{2}$. Hence, using the equality (see \cite{Le}, Eq. $(5.8.5)$, p. $112$)
\begin{equation} \label{Iundemi1}
  I_{\frac{1}{2}} (x) = \sqrt{\frac{2}{\pi x}}\ \sinh x,
\end{equation}
we find that $f_2(X,0,z) = -i\, \sinh (z(A-X))$ as expected according to Lemma \ref{Lambda=0}.
\end{remark}
\begin{remark}
  Note that the asymptotics given in (\ref{asymflun}) and  (\ref{asymfldeux}) only depend on $\kappa_+$ and $a_+$, \textit{i.e.} on some parameters of the black hole at the cosmological horizon $X=A$ (see (\ref{SurfaceGravity}) and (\ref{apm})). This is natural since the Jost functions $\fun$ and $f_2(X,\lambda,z)$ are solutions of (\ref{SL1}) with  boundaries conditions at $X=A$ and since we work in the range $X>X_0>0$, \textit{i.e.} far from the event horizon. We emphasize then that the geometry of the black hole at the event horizon $X=0$ does not appear in these asymptotics. Of course it won't be the case for the scattering data. We also stress the fact that we have only assumed that $X<X_1<A$ by convenience: the asymptotics of the derivative of the Jost functions $f_j(X,\lambda,z)$ are simpler under this condition.
\end{remark}

In order to obtain the asymptotics of the scattering data, we need to calculate the asymptotics of the Jost functions $g_j(X,\lambda,z)$. Since the procedure is the same as the one for the $f_j(X,\lambda,z)$, we give without proof the main steps to obtain the asymptotics of $g_j(X,\lambda,z), \ j=1,2$, when $z \rightarrow +\infty$. Since $g_j (X,\lambda,z)$ satisfies (\ref{SL1}) with a boundary condition at $X=0$, we work with an other diffeomorphism, denoted by $h_-(X)$, in order to construct the functions $g_j^- (X,\lambda,z)$ that approximate $g_j (X,\lambda,z)$. This new diffeomorphism is defined as follows :
\begin{equation}\label{diffeomoins}
h_-(X) = \frac{1}{\kappa_-} \ \log X + C_- ,
\end{equation}
where
\begin{equation}\label{constantemoins}
C_- =\frac{1}{\kappa_-}\ \log \left(\frac{\kappa_-}{a_-}\right).
\end{equation}
As previously, we can calculate explicitely $g_j^- (X,\lambda,z)$ and we easily obtain the following equalities :
$$
g_1^- (X,\lambda,z) =  (\frac{\kappa_-}{a_-})^{\frac{i\lambda}{\kappa_-}} \ \sqrt{X} \ \Gamma(1-\nu_-) \ (\frac{z}{2})^{\nu_-} \ I_{-\nu_- } (zX)\ ,
$$
where
\begin{equation}\label{numoins}
\nu_- = \frac{1}{2} -i \frac{\lambda}{\kappa_-}\ ,
\end{equation}
and
$$
g_2^- (X,\lambda,z) = i\ (\frac{\kappa_-}{a_-})^{-\frac{i\lambda}{\kappa_-}} \ \sqrt{X} \  \Gamma(1-\mu_-) \ (\frac{z}{2})^{\mu_-} \ I_{1-\mu_- } (zX)\ ,
$$
where
\begin{equation}\label{mumoins}
\mu_- = \frac{1}{2} +i \frac{\lambda}{\kappa_-}\ .
\end{equation}

\begin{remark}
  In the special case $\lambda=0$, we have $g_j (X,0,z) = g_j^- (X,0,z)$ and $\mu_- = \nu_- = \frac{1}{2}$. Using the formulae (\ref{Iundemi}) and (\ref{Iundemi1}) for the modified Bessel functions of half-integral order, we find $ g_1^-(X,0,z) =  \cosh (zX)$ and $g_2^-(X,0,z) = i\ \sinh (zX)$ (see Remark \ref{Lambda=0}).
\end{remark}

The $g_j^- (X,\lambda,z)$ are perturbations of the $g_j(X,\lambda,z)$. Precisely, we have

\begin{lemma}\label{gdeuxmoins}
For $X_1 \in ]0,A[$ fixed, there exists $C>0$ such that $\forall k=0,1$, $\forall X\in]0, X_1[$, $\forall z >0$,
\begin{equation}\label{diff}
\mid g_j^{(k)}(X,\lambda, z) - g_j^{- (k)} (X,\lambda, z) \mid \ \leq \ C \ X \ z^{k-1} \ e^{zX}
\end{equation}
\end{lemma}
Then, using the asymptotics of the modified Bessel functions (\ref{asymptbessel}), we obtain :

\begin{prop}\label{asymptoticgj}
For fixed $X_0<X_1$ with $X_j \in ]0, A[$, $\forall k=0,1$, $\forall X\in]X_0,X_1[$, we have the following asymptotics,
when $z \rightarrow +\infty$,
\begin{equation} \label{asymgun}
 g_1^{(k)}(X,\lambda, z)= \frac{2^{-\nu_-}}{\sqrt{2\pi}}\ (\frac{\kappa_-}{a_-})^{\frac{i\lambda}{\kappa_-}} \ \Gamma(1-\nu_-)
\ z^{k-\frac{i\lambda}{\kappa_-}} \ e^{zX} \ \Big(1+O(\frac{1}{z})\Big),
\end{equation}
\begin{equation} \label{asymgdeux}
 g_2^{(k)}(X,\lambda, z)= i \ \frac{2^{-\mu_-}}{\sqrt{2\pi}}\ (\frac{\kappa_-}{a_-})^{-\frac{i\lambda}{\kappa_-}} \ \Gamma(1-\mu_-)
\ z^{k+\frac{i\lambda}{\kappa_-}} \ e^{zX} \ \Big(1+O(\frac{1}{z})\Big),
\end{equation}
\end{prop}

\subsection{Asymptotics of the scattering data.}

In this section, we put together all the previous results and calculate the asymptotics of $a_{Lj}(\lambda,z), \ j=1,...,4$ when $z \rightarrow +\infty$. First, we recall that for all $x \in \R$,
\begin{equation}\label{rappel}
F_L (x,\lambda, z) = F_R (x,\lambda, z)\ A_L(\lambda, z).
\end{equation}
Calculating (\ref{rappel}) components by components, it follows that (in the variable $X$)
$$
f_1(X,\lambda,z) = \alun \ g_1(X,\lambda,z) + \alt \ g_2(X,\lambda,z)
$$
$$
f_2(X,\lambda,z) = \ald \ g_1(X,\lambda,z) + \alq \ g_2(X,\lambda,z)
$$
So, by Lemma \ref{wronskien}, we obtain for $z\not=0$ :
\begin{eqnarray*}
\alun &=& \frac{1}{iz} \ W(f_1,g_2)\ ,\quad \ald = \frac{1}{iz} \ W(f_2,g_2),\\
\alt &=& - \frac{1}{iz} \ W(f_1,g_1)\ ,\quad \alq = - \frac{1}{iz} \ W(f_2,g_1).
\end{eqnarray*}
The following theorem is an easy consequence of Propositions \ref{asymptoticfun}, \ref{asymptoticfdeux} and \ref{asymptoticgj} :

\begin{theorem}\label{asymptoticsal}
When $z \rightarrow + \infty$, we have :
\begin{eqnarray}
\alun & \sim & \ \frac{1}{2\pi}\ \left(-\frac{\kappa_+}{a_+}\right)^{\frac{i\lambda}{\kappa_+}}
\left(\frac{\kappa_-}{a_-}\right)^{-\frac{i\lambda}{\kappa_-}}
\Gamma\left(\frac{1}{2}-\frac{i\lambda}{\kappa_-}\right) \Gamma\left(\frac{1}{2}+\frac{i\lambda}{\kappa_+}\right)  \
\left(\frac{z}{2}\right)^{i\lambda (\frac{1}{\kappa_-} - \frac{1}{\kappa_+})}  e^{zA} \\
\ald & \sim & \  \frac{-i}{2\pi}\ \left(-\frac{\kappa_+}{a_+}\right)^{-\frac{i\lambda}{\kappa_+}}
\left(\frac{\kappa_-}{a_-}\right)^{-\frac{i\lambda}{\kappa_-}}
\Gamma\left(\frac{1}{2}-\frac{i\lambda}{\kappa_-}\right) \Gamma\left(\frac{1}{2}-\frac{i\lambda}{\kappa_+}\right)  \
\left(\frac{z}{2}\right)^{i\lambda (\frac{1}{\kappa_+} + \frac{1}{\kappa_-})}  e^{zA} \\
\alt & \sim & \ \frac{i}{2\pi}\ \left(-\frac{\kappa_+}{a_+}\right)^{\frac{i\lambda}{\kappa_+}}
\left(\frac{\kappa_-}{a_-}\right)^{\frac{i\lambda}{\kappa_-}}
\Gamma\left(\frac{1}{2}+\frac{i\lambda}{\kappa_-}\right) \Gamma\left(\frac{1}{2}+\frac{i\lambda}{\kappa_+}\right)  \
\left(\frac{z}{2}\right)^{-i\lambda (\frac{1}{\kappa_+} + \frac{1}{\kappa_-})}  e^{zA} \\
\alq & \sim & \ \frac{1}{2\pi}\ \left(-\frac{\kappa_+}{a_+}\right)^{-\frac{i\lambda}{\kappa_+}}
\left(\frac{\kappa_-}{a_-}\right)^{\frac{i\lambda}{\kappa_-}}
\Gamma\left(\frac{1}{2}+\frac{i\lambda}{\kappa_-}\right) \Gamma\left(\frac{1}{2}-\frac{i\lambda}{\kappa_+}\right)  \
\left(\frac{z}{2}\right)^{i\lambda (\frac{1}{\kappa_+} - \frac{1}{\kappa_-})}  e^{zA}
\end{eqnarray}
\end{theorem}
We deduce from Theorem \ref{asymptoticsal} the asymptotics of the transmission and reflexion coefficients $T(\lambda,z)$, $L(\lambda,z)$ and $R(\lambda,z)$. From the definitions (\ref{SR-SM2}) of these coefficients, we get
\begin{theorem}\label{asymptoticscat}
When $z \rightarrow + \infty$, we have :
\begin{eqnarray*}
T(\lambda,z) & \sim & 2\pi \left(-\frac{a_+}{\kappa_+}\right)^{\frac{i\lambda}{\kappa_+}}
\left(\frac{a_-}{\kappa_-}\right)^{-\frac{i\lambda}{\kappa_-}} \frac{1}{\Gamma\left(\frac{1}{2}-\frac{i\lambda}{\kappa_-}\right) \Gamma\left(\frac{1}{2}+\frac{i\lambda}{\kappa_+}\right)}  \
\left(\frac{z}{2}\right)^{i\lambda (\frac{1}{\kappa_+} - \frac{1}{\kappa_-})}  e^{-zA}, \\
L(\lambda,z) & \sim & i\ \left(\frac{\kappa_-}{a_-}\right)^{\frac{2i\lambda}{\kappa_-}}\
\frac{\Gamma\left(\frac{1}{2}+\frac{i\lambda}{\kappa_-}\right)}{\Gamma\left(\frac{1}{2}-\frac{i\lambda}{\kappa_-}\right)}\
\left(\frac{z}{2}\right)^{-\frac{2i\lambda}{\kappa_-}} \ .\\
R(\lambda,z) & \sim & i\ \left(-\frac{\kappa_+}{a_+}\right)^{-\frac{2i\lambda}{\kappa_+}}\
\frac{\Gamma\left(\frac{1}{2}-\frac{i\lambda}{\kappa_+}\right)}{\Gamma\left(\frac{1}{2}+\frac{i\lambda}{\kappa_+}\right)}\
\left(\frac{z}{2}\right)^{\frac{2i\lambda}{\kappa_+}} \ .
\end{eqnarray*}
\end{theorem}

\begin{remark}
  As expected, the asymptotic of the transmission coefficient $T(\lambda,z)$ depends on the parameters $\kappa_\pm$ and $a_\pm$, \textit{i.e.} on the geometries of both event and cosmological horizons. On the other hand, the asymptotic of the reflection coefficient $L(\lambda,z)$ depends only on $\kappa_-$ and $a_-$ - the geometry of the event horizon - whereas the asymptotic of $R(\lambda,z)$ depends on $\kappa_+$ and $a_+$ - the geometry of the cosmological horizon.
\end{remark}


\subsection{Reconstruction formulae for $\kappa_\pm$}

\vspace{0,1cm}\noindent
As a by-product of the asymptotics obtained in Theorem \ref{asymptoticscat}, we find simple reconstruction formulae for the surface gravities $\kappa_{\pm}$ from the scattering reflexion coefficients $L(\lambda,n)$ and $R(\lambda,n)$, $n \in \N$. As already mentioned at the beginning of this section, these quantities are meaningful in the Hawking effect.
\begin{theorem}\label{reconstruction} For all $p\in \N$, we have :
\begin{eqnarray*}
\lim_{n \rightarrow +\infty} \ \frac{L(\lambda,pn)}{L(\lambda,n)} \ &=& \ e^{-\frac{2i\lambda}{\kappa_-}\log p}\ ,\\
\lim_{n \rightarrow +\infty} \ \frac{R(\lambda,pn)}{R(\lambda,n)} \ &=& \ e^{\frac{2i\lambda}{\kappa_+}\log p}\ .
\end{eqnarray*}
\end{theorem}
We can now determine easily the surface gravities for nonzero energies $\lambda \not= 0$. For example, if we set :
\begin{equation}\label{suite}
u_p \ =\ e^{-\frac{2i\lambda}{\kappa_-}\log p},
\end{equation}
which is known, we obtain when $p \rightarrow +\infty$ :
\begin{equation}\label{dl}
\frac{u_{p+1}}{u_p} = 1 -\frac{2i\lambda}{\kappa_- p} +O(\frac{1}{p^2}).
\end{equation}
This permits to calculate $\kappa_-$ if $\lambda\not=0$.


\Section{The inverse scattering problem.} \label{Inverse}

In this section, we prove our main result Thm \ref{Main}, that is we prove the uniqueness of the potential $a(x)$ up to translations as well as the uniqueness of the parameters $(M,Q^2, \Lambda)$ of a dS-RN black hole from the knowledge of either the transmission coefficient $T(\lambda,n)$, or the reflection coefficients $L(\lambda,n)$ or $R(\lambda,n)$ at a {\it{fixed energy}} $\lambda\not= 0$ and for all $n \in \mathcal{L} \subset \N^*$ satisfying the M\"untz condition $\sum_{n \in \mathcal{L}} \frac{1}{n} = \infty$.

Consider thus two dS-RN black holes with parameters $(M,Q,\Lambda)$ and $(\tilde{M},\tilde{Q},\tilde{\Lambda})$ respectively.
We shall denote by $a(x)$ and $\tilde{a}(x)$ the corresponding potentials appearing in the Dirac equation and satisfying the hypotheses of section 2.
We shall also use the notation $\tilde{Z}$ for all the scattering data associated with the potential $\tilde{a}$. As explained in the introduction, we assume that there exists a constant $c$ such that one of the following equalities hold for all $n \in \cal{L}$\footnote{Recall that we add a constant $c$ in (\ref{egalitecoefS}) to include the possiblity of describing the same dS-RN black hole by two different RW variables and make our result coordinates independent.}
\begin{equation}\label{egalitecoefS}
\left\{ \begin{array}{ccc}
T(\lambda, n) &=& \tilde{T}(\lambda, n),  \\
L(\lambda, n) &=& e^{-2i\lambda c}\ \tilde{L}(\lambda, n), \\
R(\lambda, n) &=& e^{2i\lambda c}\ \tilde{R}(\lambda, n).
\end{array} \right.
\end{equation}
By Propositions \ref{Uniqueness} and \ref{Uniqueness2}, we deduce from (\ref{egalitecoefS}) that
\begin{eqnarray}
  a_{L1}(\lambda, z) = \tilde{a}_{L1} (\lambda, z) \ ,\quad a_{L2}(\lambda, z) = e^{2i\lambda c} \tilde{a}_{L2}(\lambda, z), \label{egalite0} \\
a_{L3}(\lambda, z) = e^{-2i\lambda c} \ \tilde{a}_{L3} (\lambda, z) \ , \quad
a_{L4}(\lambda, z) = \tilde{a}_{L4}(\lambda, z). \label{egalite}
\end{eqnarray}
Thus, it follows from the asymptotics of Theorem \ref{asymptoticsal} that :
\begin{equation}\label{uniciteA}
A := \int_{-\infty}^{+\infty} \ a (x) \ dx \ =\ \int_{-\infty}^{+\infty} \ \tilde{a} (x) \ dx \  = \tilde{A}.
\end{equation}
Hence, we can define the diffeomorphisms $h, \ \tilde{h} : \ ]0,A[ \rightarrow \R$ as the inverses of the Liouville transforms $g$ and $\tilde{g}$ given by (\ref{Liouville}) in which we use the potentials $a(x)$ and $\tilde{a}(x)$ respectively.

Now, following a strategy relatively close to \cite{FY}, we introduce, for $X \in ]0,A[$, the matrix
$$
  P(X,\lambda,z) = \left( \begin{array}{cc} P_1(X,\lambda,z) & P_2(X,\lambda,z) \\
                                            P_3(X,\lambda,z) & P_4(X,\lambda,z)
                                            \end{array} \right),
$$
defined by
\begin{equation}\label{matricepassage}
P(X,\lambda,z) \ \tilde{F}_R (\tilde{h} (X), \lambda, z) \ = \ F_R (h(X), \lambda, z)\ e^{i\lambda c \Gamma^1},
\end{equation}
where $F_R =(f_{Rk})$ and $\tilde{F}_R =(\tilde{f}_{Rk})$ are the Jost solutions from the right associated with $a(x)$ and $\tilde{a}(x)$. To simplify the notation, for $k=1, ...,4$, we set as in Section 4:
\begin{eqnarray*}
f_{k} (X,\lambda, z) = f_{Lk} (h (X), \lambda, z), & \tilde{f}_{k} (X,\lambda, z) = \tilde{f}_{Lk} (\tilde{h} (X), \lambda, z),\\
g_{k} (X,\lambda, z) = f_{Rk} (h (X), \lambda, z), & \tilde{g}_{k} (X,\lambda, z) = \tilde{f}_{Rk} (\tilde{h} (X), \lambda, z).
\end{eqnarray*}
Using that det $F_R =1$ and det $\tilde{F}_R = 1$, we obtain the following equalities :
\begin{equation}\label{premiereformule}
\left\{ \begin{array}{ccc}
P_1(X,\lambda, z) &=& e^{i\lambda c} \ g_{1}\  \tilde{g}_{4} - e^{-i\lambda c} \ g_{2} \ \tilde{g}_{3},  \\
P_2(X,\lambda, z) &=& - e^{i\lambda c} \ g_{1} \ \tilde{g}_{2} - e^{-i\lambda c} \ g_{2} \ \tilde{g}_{1} .
\end{array}
\right.
\end{equation}
It follows from (\ref{premiereformule}) and the analytical properties of the Jost functions that, for $j=1,2$,
the applications $z \rightarrow P_j(X,\lambda,z)$ are holomorphic on $\C$ and of exponential type. Moreover, by Lemma \ref{MainEstiF}, these applications are bounded on the imaginary axis $i \R$.

We shall now prove that the applications $z \rightarrow P_j(X,\lambda,z)$ are also bounded on the real axis. To do this, we first make some elementary algebraic transformations on $P_j(X,\lambda,z)$. We write :
\begin{eqnarray*}
  P_1(X,\lambda,z) & = & e^{i\lambda c} \ g_{1}\  g_{4} + e^{i\lambda c} \ g_{1}\  (\tilde{g}_{4}-g_{4})
                   -\ e^{-i\lambda c} \ g_{2} \ g_{3} - e^{-i\lambda c} \ g_{2}\  (\tilde{g}_{3}-g_{3}), \\
                   & = & e^{i\lambda c} \ (g_{1}\  g_{4} - g_{2}\  g_{3})
                      - \ e^{-i\lambda c} \ g_{2}\  (g_{3}- e^{2i\lambda c}\ g_{3}) \\
                   &   & \ \ \ +\ e^{i\lambda c} \ g_{1}\  (\tilde{g}_{4}-g_{4})
                      - \ e^{-i\lambda c} \ g_{2}\  (\tilde{g}_{3}-g_{3}), \\
                   & = & e^{i\lambda c} +  e^{i\lambda c} \ g_{2}\ g_{3} - e^{-i\lambda c} \ g_{2}\ \tilde{g}_{3}
                      +\  e^{i\lambda c} \ g_{1}\  (\tilde{g}_{4}-g_{4}),
\end{eqnarray*}
where we have used that det $F_R= g_{1}\  g_{4} - g_{2}\  g_{3} =1$. Since $F_L(x,\lambda,z)=F_R (x, \lambda, z)\ A_L (\lambda, z)$, we get using (\ref{egalite}) :
$$ 
g_{4} = \frac{1}{a_{L4}} \ (f_{4} - a_{L2}\ g_{3}), \quad \quad \tilde{g}_{4} = \frac{1}{a_{L4}} \ (\tilde{f}_{4} - e^{-2i\lambda c}\ a_{L2}\ \tilde{g}_{3}).
$$
So, we obtain immediately :
$$
P_1(X,\lambda,z) = e^{i\lambda c} +  e^{i\lambda c}\ \frac{g_{1}}{a_{L4}} \ (\tilde{f}_{4}- f_{4})
                 + (g_{2} + \frac{g_{1} }{a_{L4}} \ a_{L2} ) \ (e^{i\lambda c} \ g_{3}-e^{-i\lambda c} \ \tilde{g}_{3}).
$$
Using again $F_L(x,\lambda,z)=F_R (x, \lambda, z)\ A_L (\lambda, z)$, we see that $f_{2} = a_{L2}\ g_{1} + a_{L4} \ g_{2}$. Thus, we get
\begin{equation}\label{secondeformuleun}
P_1(X,\lambda,z) = e^{i\lambda c} +  e^{i\lambda c}\ \frac{g_{1}}{a_{L4}} \ (\tilde{f}_{4}- f_{4})
                 + \frac{f_{2}}{a_{L4}} \ (e^{i\lambda c} \ g_{3}-e^{-i\lambda c} \ \tilde{g}_{3}).
\end{equation}
Similarly, $P_2(X,\lambda,z)$ can be expressed as :
\begin{equation}\label{secondeformuledeux}
P_2(X,\lambda,z) = \frac{1}{a_{L4}} \ \left( e^{-i\lambda c} \ f_{2}\ \tilde{g}_{1} - e^{i\lambda c} \ \tilde{f}_{2}\ g_{1} \right).
\end{equation}

We shall now use some estimates obtained in the previous sections. First, it follows from Lemma \ref{MainEstiF} that for $z>0$ and for all $j=1,..,4$ :
\begin{equation}\label{rappels}
\mid f_{j}(X,\lambda,z) \mid, \ \mid \tilde{f}_{j}(X,\lambda,z) \mid \leq e^{z(A-X)} \ ,\quad \quad \mid g_{j}(X,\lambda,z) \mid, \ \mid \tilde{g}_{j}(X,\lambda,z) \mid\leq e^{zX}.
\end{equation}
Second, using Theorem \ref{asymptoticsal}, it is easy to see that for $z$ real and large enough
\begin{equation} \label{Minoration}
  |a_{L4}(\lambda,z)| \geq C e^{Az}, \quad z>>1.
\end{equation}
Hence, using (\ref{secondeformuleun}), (\ref{secondeformuledeux}), (\ref{rappels}) and (\ref{Minoration}), we conclude that
for all fixed $X\in ]0,A[$, the applications $z \rightarrow P_j (X,\lambda,z)$ are bounded on $\R^+$. Clearly, this result remains true on $\R$ by an elementary parity argument. Finally, applying the Phragmen-Lindel\" of's theorem (\cite{Bo}, Thm 1.4.2.) on each quadrant of the complex plane, we deduce that $z \rightarrow P_j (X,\lambda,z)$ is bounded on $\C$. By
Liouville's theorem, we have thus obtained :
\begin{equation}\label{thliouville}
P_j (X,\lambda,z)=P_j (X,\lambda,0) \ \ ,\ \ \forall z \in \C .
\end{equation}

Now, we return to the definition of $P_j(X,\lambda,z)$ for $z=0$. We observe first that ${\displaystyle{F_R (x,\lambda, 0) = e^{i \lambda \Gamma^1 x} }} $ and similarly ${\displaystyle{\tilde{F}_R (x,\lambda, 0) = e^{i \lambda \Gamma^1 x} }} $. This is immediate from the definition of the Jost function. Thus we deduce from (\ref{matricepassage}) that
\begin{equation}\label{egalitemat}
P(X,\lambda, 0) = e^{i\  \lambda \ ( h(X)-\tilde{h}(X)+c)\ \Gamma^1 }.
\end{equation}
Then, putting (\ref{egalitemat}) and (\ref{thliouville}) into (\ref{matricepassage}) we get
\begin{equation}\label{egalitejost}
\left\{ \begin{array}{ccc}
\tilde{g}_{1}(X,\lambda, z) &=& e^{i \theta(X)}  \ g_{1}(X,\lambda, z),   \\
\tilde{g}_{2}(X,\lambda, z) &=& e^{-2 i\lambda c}\ e^{i \theta(X)} \ g_{2}(X,\lambda, z),
\end{array}
\right.
\end{equation}
where we have set $\theta(X) = \lambda \ (\tilde{h}(X) - h (X))$.

By Lemma \ref{wronskien}, the wronskians $W(g_{1} , g_{2}) = W(\tilde{g}_{1}  , \tilde{g}_{2}) = iz$. Then, a straightforward calculation gives
\begin{equation}\label{egalitephase}
e^{2i (\theta(X)- \lambda c)} \ =\ 1.
\end{equation}
Thus, by a standard continuity argument, there exists $k \in \Z$ such that
\begin{equation}\label{thetaX}
\theta (X)= \lambda c + k \pi \ \ ,\ \ \forall X \in ]0,A[,
\end{equation}
or equivalently
\begin{equation}\label{diffeos}
\tilde{h}(X)= h(X) + c + \frac{k \pi}{\lambda} \ \ ,\ \ \forall X \in ]0,A[,
\end{equation}
Let us differentiate (\ref{diffeos}) with respect to $X$. We obtain easily
\begin{equation}
\frac{1}{a(\tilde{h}(X))} = \frac{1}{a(h(X))},
\end{equation}
and using again (\ref{diffeos}), we have
\begin{equation} \label{UnicitePot}
  a(x) = \tilde{a}(x + c + \frac{k \pi}{\lambda}) \ \ , \ \  \forall x \in \R.
\end{equation}
Thus, we have proved the first part of Theorem \ref{Main}.

\vspace{0.5cm}
We are now in position to finish the proof of Thm \ref{Main} and prove the uniqueness of the mass
$M$, the square of the charge $Q^2$ and the cosmological constant $\Lambda$ of the black hole. First, recall that
\begin{equation}\label{acarre}
a^2(x) = \frac{F(r)}{r^2} = \frac{1}{r^2} - \frac{2M}{r^3} + \frac{Q^2}{r^4} - \frac{\Lambda}{3},
\end{equation}
where $r$ stands for $r(x)$ the inverse of the Regge-Wheeler diffeomorphism\footnote{We emphasize here that $r(x)$ depends on the parameters we are looking for.}.

To prove the uniqueness of the parameters, we use the following trick. We define the differential operator $B$ by :
\begin{equation}
B = \frac{1}{a^2(x)} \ \frac{d}{dx} = r^2 \ \frac{d}{dr},
\end{equation}
since ${\displaystyle{\frac{dr}{dx} = F(r)}}$. Using the notation $B^2 = B \circ B$, etc..., a straightforward calculation gives
\begin{eqnarray}\label{itereeundeB}
B(a^2) &=& -\frac{2}{r} + \frac{6\ M}{r^2} - \frac{4\ Q^2}{r^3}.  \\
\label{itereedeuxdeB}
B^2(a^2) &=& 2 - \frac{12\ M}{r} + \frac{12\ Q^2}{r^2}.  \\
\label{itereetroisdeB}
B^3(a^2) &=& 12\ M -  \frac{24\ Q^2}{r}. \\
\label{itereequatredeB}
B^4(a^2) &=& 24\ Q^2.
\end{eqnarray}
Now setting $\tilde{x} =x+c+ \frac{k \pi}{\lambda}$ and using (\ref{UnicitePot}), we remark that :
\begin{equation}
B = \frac{1}{\tilde{a}^2(\tilde{x})} \ \frac{d}{d\tilde{x}} = \frac{1}{a^2(x)} \ \frac{d}{dx}.
\end{equation}
We apply this differential operator to the equality $\tilde{a}(\tilde{x})^2 = a(x)^2$. To simplify the notation, we set
$\tilde{r} = \tilde{r}(\tilde{x})$ and $r= r (x)$.

Using (\ref{itereequatredeB}) and (\ref{itereetroisdeB}),  we obtain successively :
\begin{eqnarray}\label{uniciteQ}
Q^2&=&\tilde{Q}^2, \\
\frac{1}{\tilde{r}} - \frac{1}{r} &=& \frac{\tilde{M}-M}{2Q^2} := E.
\end{eqnarray}
Then, using (\ref{itereedeuxdeB}), we have
\begin{equation}\label{uniciteM}
\frac{M + \tilde{M}}{2} \ E = (\tilde{M} - E Q^2)\ E = \frac{M - \tilde{M} + 2Q^2 E}{r} = 0.
\end{equation}
So, we deduce from (\ref{uniciteM}) that $E=0$ since $M,\ \tilde{M}>0$, \textit{i.e.} we have obtained $M =
\tilde{M}$ and $r = \tilde{r}$. Using now (\ref{acarre}), we get $\Lambda = \tilde{\Lambda}$ and the proof is
complete. $\diamondsuit$

\vspace{0.5cm}\noindent
{\it{Acknowledgments.}}
\par\noindent
This work was initiated while F.N was visiting T.D. at McGill University. Both authors would like to warmly thank Niky Kamran for his hospitality and encouragement.


\end{document}